\documentclass[10pt,a4paper, final]{IEEEtran}

\usepackage{setspace}

\usepackage[pdftex]{graphicx}

\usepackage{graphicx}
\usepackage{subcaption}
\usepackage{mwe}

\newsavebox{\measurebox}

\usepackage{tikz}

\definecolor{gold}{rgb}{0.85,.66,0}

\usepackage{cite}
\usepackage{url}

\usepackage{amsthm}
\usepackage{amssymb}
\usepackage{amsmath}
\usepackage[cmintegrals]{newtxmath}

\usepackage{comment}

\usepackage{bm}

\makeatletter
\def\munderbar#1{\underline{\sbox\tw@{$#1$}\dp\tw@\z@\box\tw@}}
\makeatother

\newtheorem{theorem}{Theorem}

\newtheorem{remark}{Remark}
\newtheorem{definition}{Definition}

\usepackage{soul}
\usepackage[normalem]{ulem} 

\usepackage{algorithm}
\usepackage{algpseudocode}

\ifCLASSINFOpdf
\else
\fi

\begin{document}



\newcommand{\papertitle}{
    User-Centric Perspective in Random Access Cell-Free Aided by Spatial Separability
}

\title{\papertitle}

\author{
    {Victor Croisfelt},
    {Taufik Abrão}, and
    {José Carlos Marinello}\\
    \vspace{5mm}
    {\small
    \textbf{NOTE:} This work has been submitted to the IEEE for possible publication.\\
    Copyright may be transferred without notice, after which this version may no longer be accessible.
    }
    \thanks{V. Croisfelt was with the Electrical Engineering Department, Universidade de
    São Paulo, Escola Politécnica, São Paulo, Brazil Now, he is with the Connectivity Section of the Department of Electronic Systems, Aalborg University, Aalborg, Denmark; E-mail: \texttt{vcr@es.aau.dk}} 
    \thanks{{T. Abrão is with the Electrical Engineering Department, State University of Londrina, PR, Brazil.  E-mail: \texttt{taufik@uel.br}}}%
    \thanks{J. C. Marinello is with the Electrical Engineering Department, Universidade Tecnológica Federal do Paraná, Cornélio Procópio, PR, Brazil.  E-mail: \texttt{jcmarinello@utfpr.edu.br}}%
\vspace{-5mm}}


\maketitle

\begin{abstract} 
    In a cell-free massive multiple-input multiple-output (CF-mMIMO) network, multiple access points (APs) actively cooperate to serve users' equipment (UEs). We consider how the random access (RA) problem can be addressed by such a network under the occurrence of pilot collisions. To find a solution, we embrace the user-centric perspective, which basically dictates that only a preferred set of APs needs to serve a UE. Due to the success of the strongest-user collision resolution (SUCRe) protocol for cellular (Ce) mMIMO, we  extend it by considering the new setting. {Besides,} we establish that the user-centric perspective naturally equips a CF network with robust fundamentals for resolving collisions. We refer to this foundation as spatial separability, which enables multiple colliding UEs to access the network simultaneously. We then propose two novel RA protocols for CF-mMIMO: i) the baseline cell-free (BCF) that resolves collisions with the concept of spatial separability alone, and ii) the cell-free SUCRe (CF-SUCRe) that combines SUCRe and spatial separability principle to resolve collisions. We evaluate our proposed RA protocols against the Ce-SUCRe. Respectively, the BCF and CF-SUCRe can support 7$\times$ and 4$\times$ more UEs' access on average compared to the Ce-SUCRe with an average energy efficiency gain based on total power consumed (TPC) by the network per access attempt of 52$\times$ and 340$\times$. Among our procedures, even with a higher overhead, the CF-SUCRe is superior to BCF regarding TPC per access attempt. This is because the combination of methods for collision resolution allows many APs to be disconnected from the RA process without sacrificing much the performance. {Finally, our numerical results 
    can be reproduced using the code package available on: \emph{github.com/victorcroisfelt/cf-ra-spatial-separability}.} 
\end{abstract}

\begin{IEEEkeywords}
Cell-free; Massive MIMO; User-centric; Random access; Grant based protocols; Spatial separability; 6G systems
\end{IEEEkeywords}

%
\IEEEpeerreviewmaketitle

\section{Introduction}\label{sec:intro}
\IEEEPARstart{I}{n} future mobile networks, massive wireless connectivity is a fundamental requirement sought to support the expected huge amount of users' equipment (UEs) \cite{Popovski2020,Bana2019}. Thereby, the task of improving the performance of random access (RA) to the network has been the focus of considerable research \cite{Bana2019}. Recently, cellular massive multiple-input multiple-output (Ce-mMIMO) technology has proven to be a great ally in achieving these improvements with, for example, the introduction of the strongest-user collision resolution (SUCRe) protocol in \cite{Bjornson2017}. SUCRe exploits the channel hardening and favorable propagation capabilities of Ce-mMIMO systems, allowing for a distributed resolution in each of the UEs of collisions that occur due to pilot shortage. SUCRe's success resulted in the proposition of several variants that seek to solve problems of the original version from \cite{Bjornson2017}, 
such as the unfairness caused by the fact that UEs closer to the base station (BS) are more likely to be favored by the protocol \cite{Marinello2019}.

{With the increasing challenge of pushing the boundaries of connectivity beyond the offering of high wireless data rates and expand towards interconnecting humans, machines, robots, and things \cite{popovski2021IoT}, new technologies has been gaining attention in addition to Ce-mMIMO. In fact, as the most diverse types of services rely more on the network, we need to raise the \emph{uniformity} of how these services gain access to the network and the \emph{availability} (coverage) of this access everywhere. Because of the intelligent use of spatial diversity, Ce-mMIMO partially provides the uniformity and ubiquity requirements; however, the imbalance persists due to interference and {the low performance delivered to UEs located at the cell-edges} \cite{Demir2021}, which are intrinsic problems of {the cellular perspective and, hence,} of SUCRe \cite{Bjornson2017}. Following this line,} the authors of \cite{Carvalho2020} introduced the concept of extra-large-scale MIMO (XL-MIMO) by arguing that Ce-mMIMO operates differently when the size of the antenna arrays is very large{, bringing some performance benefits}. Remarkably, the authors of \cite{Nishimura2020} adapted the SUCRe protocol for XL-MIMO systems; {while \cite{Marinello2021} better exploits the concept of visibility regions to separate UEs}. 

Another promising approach is to completely get rid of the cellular perspective of design, giving rise to the notion of cell-free (CF) mMIMO networks \cite{Demir2021}. In these systems, a UE is surrounded by now called access points (APs) that actively cooperate with each other and also make use of the spatial principles behind mMIMO technology. In particular, we are interested in the \emph{user-centric perspective} of CF networks \cite{Demir2021}, where a UE is served by a preferred subset of APs {based on its needs. With the selection and jointly operation of the preferred APs, the goal of a more uniform and ubiquitous network can be better addressed \cite{Demir2021}. Herein, we exploit the design of a CF-mMIMO network to raise the uniformity and ubiquity of the RA.} Since CF-mMIMO has been gaining interest only recently, {there are few works considering the} RA framework under such scenario \cite{Ganesan2020,Henriksson2020,Chen2021,Ganesan2021}.

In this work, we propose an extension of the SUCRe protocol for CF-mMIMO networks. Our motivation comes from the lack of protocols in the literature that adequately exploit the user-centric CF design {for RA}. Furthermore, we emphasize the interest in grant-based (GB) RA protocols, which are justified when the size of the information to be transmitted by UEs is large enough, like in crowded enhanced mobile broadband services \cite{Bjornson2017}, {and other important Internet of Things (IoT) applications, such as street video monitors, vehicle-to-vehicle communications, vehicle auto-diagnosis and autopilot, remote surgery, and smart-home/enterprise services, including video camera, TV, laptop, and printer \cite{Chen17}}. Alternatively, grant-free (GF) protocols are intended for application scenarios in which the communication of UEs is very sporadic and short, {typically associated to} massive machine-type communications {(mMTC) applications such as wireless sensor networks and smart monitoring \cite{Bana2019}}.

\subsection{{Literature Review}}
There are few works that discuss the RA problem in CF-mMIMO networks \cite{Ganesan2020,Ganesan2021,Henriksson2020,Chen2021}. In \cite{Ganesan2020} and \cite{Ganesan2021}, the authors proposed a GF protocol based on activity detection using the maximum likelihood method, showing the first traces of RA performance gains with the CF architecture. These improvements are intrinsically obtained by exploiting the \emph{augmented macro-diversity} owing to the existence of diverse geographically distributed APs. Of course, the price to pay for this macro-diversity is the implementation of further infrastructure to support the CF architecture. Following the same line of \cite{Ganesan2020}, the {works in} \cite{Henriksson2020} and \cite{Chen2021} use other mathematical frameworks to introduce more efficient algorithms to obtain a GF protocol considering activity detection. Contrastingly, our work considers the design of a GB protocol for CF-mMIMO towards collision resolution. 

{In \cite{Ding2021}, the authors examine the distinctive features and benefits of CF-mMIMO to resolve GF ultra-reliable and low-latency (URLLC) communication issues. They shown that distinctive features of CF-mMIMO can be deployed to resolve preamble collision, suppressing multiuser interference in GF RA. Open issues and challenges in GF CF-mMIMO include: a) preamble-collision resolution, b) channel inference assisted resource allocation, and c) limited fronthaul. Preamble-collision resolution in GF URLLC could be achieved by coded RA and its variants in the context of HARQ retransmission schemes, relying on the use of multiple preambles over transmission time intervals (TTIs). Hence, by exploiting macro-diversity and signal spatial sparsity, CF-mMIMO is capable of resolving preamble collisions on the basis of a single TTI, which is beyond the capability of Ce-mMIMO \cite{Ding2021}.} This last work and \cite{Hu2021} relate to our {since they exploit} the augmented macro-diversity from the CF architecture to mitigate collisions between pilots. However, as far as we know, our work is the first to explore and analytically analyze the impact of the spatial separability tool due to the greater macro-diversity of a user-centric CF system in the design of GB RA protocols.

\subsection{Our Contributions}
In this paper, we extend the GB SUCRe protocol from \cite{Bjornson2017} for CF-mMIMO networks. This extension aims to fully exploit the user-centric perspective \cite{Demir2021}, that is, {we follow the principle that} \emph{not all APs necessarily need to serve all RA pilots and not even be operational on the RA phase as a whole}. On this basis, the contributions of this work are threefold.
\begin{itemize}
    \item The introduction of the concept of spatial separability is the core novelty of this work. Spatial separability is a tool that allows multiple colliding UEs to access the network simultaneously, even if the collision is not properly resolved. We analytically characterize the spatial separability potential of a CF-mMIMO system. 
    \item We propose two new GB-RA protocols for CF-mMIMO networks aided by the spatial separability principle: \textit{i)} the \textbf{baseline cell-free (BCF) protocol}, which just exploits the spatial separability concept to resolve collisions, and \textit{ii)} the \textbf{cell-free SUCRe (CF-SUCRe) protocol}, which combines SUCRe and spatial separability concepts, while preserves the \emph{decentralized way} of resolving collisions.
\item {To enable the SUCRe-based collision resolution in CF networks,} we introduce three new estimators for estimating the total uplink (UL) signal power from the colliding UEs.  We comprehensively evaluate these three estimators.
\end{itemize}
The salient features of the proposed CF-RA protocols include: First, we introduce the notions that: \textit{i}) a UE is only aware of a preferred set of APs and \textit{ii)} another preferred set of APs is allocated by the CF network to serve each RA pilot. 
Second, the difference between the preferred sets of APs of each UE and the preferred sets of APs servicing each pilot generates a particular macro-diversity in CF systems; we exploit this macro-diversity for collision resolution and refer to it as \emph{spatial separability}. Third, we design three new estimators taking into account that a UE only knows its preferred set of APs, which is different from the set of APs that service the RA pilot sent by it. In essence, both proposed GB RA protocols deploy the idea of spatial separability: the BCF just exploits the spatial separability concept to resolve collisions; while the CF-SUCRe combines the concepts of SUCRe and spatial separability to resolve collisions. Our main results reveal that the CF-SUCRe can perform as well as the BCF, although it is more energy efficient. This is because the combination of SUCRe and spatial separability allows to reduce the amount of APs operating in the RA admission steps without affecting the its performance substantially. Hence, we refer to the original SUCRe \cite{Bjornson2017} as the \textbf{cellular SUCRe (Ce-SUCRe) protocol}.

\subsection{Notation}
Let $\mathbb{R}_{+}$ denote the set of positive real numbers and $\mathbb{C}$ the set of complex numbers. Integer sets are denoted by calligraphic letters $\mathcal{A}$ with cardinality given by $|\mathcal{A}|$ and empty set $\emptyset$. Lower case boldface letters denote column vectors (\emph{e.g.}, $\mathbf{x}$), while uppercase boldface letters stand for matrices (\emph{e.g.}, $\mathbf{A}$). The identity matrix of size $N$ is $\mathbf{I}_N$, whereas $\mathbf{0}$ and $\mathbf{1}$ stands for a vector of zeros and ones, respectively. The Euclidean norm of an arbitrary vector $\mathbf{x}$ is $\lVert\mathbf{x}\rVert_2$, while $\lVert\mathbf{x}\rVert_1$ is its $l_1$-norm. The $\mathrm{argsort}(\cdot)$ function sorts and returns the indices of a vector in ascending order. The ceil function is $\lceil\cdot\rceil$. {The circularly-symmetric complex Gaussian distribution is denoted as $\mathcal{N}_{\mathbb{C}}(\mu,\sigma^2)$ with mean $\mu$ and variance $\sigma^2$, while ${\rm U}_{[a,b]}$ denotes the continuous uniform distribution in the range $[a, b]$.} Operators for probability and expectation are $\mathbb{P}\{\cdot\}$ and $\mathbb{E}\{\cdot\}$, respectively. In algorithm pseudo-codes, $\gets$ denotes assignment to a variable and $\mathrm{sum}(\cdot)$ defines a function that sums Boolean values, where "True" is 1 and "False" is 0. {The symbol \# means "number of".}

\section{System Model}\label{sec:systemmodel}
Suppose that a subset $\mathcal{K}\subset\mathcal{U}$ of inactive UEs requests access to a CF-mMIMO network comprised of $L$ APs at a given moment. The APs are equipped with $N$ antennas each and are indexed by $\mathcal{L}=\{1,2,\dots,L\}$. A central processing unit (CPU) coordinates the exchange of information between APs through fronthaul links. {In addition, we consider that APs and UEs are placed within a square area of $\ell^2$}. Fig. \ref{fig:system-model} illustrates the adopted system model. {We let} $P_a$ be the \emph{probability of access} that dictates whether an inactive UE $k\in\mathcal{U}$ tries to access the network or not, that is, whether $k$ is a member of $\mathcal{K}$.{\footnote{{An equal, independent probability of access $P_a$ for each UE covers a more general and unfavorable scenario: UEs have the same chance of being active, while UEs' activities are assumed uncorrelated.}}} Moreover, we assume that the RA phase occurs through the transmission of pilots. The pilot pool $\boldsymbol{\Phi}\in\mathbb{C}^{\tau_p\times\tau_p}=[\boldsymbol{\phi}_1,\boldsymbol{\phi}_2,\dots,\boldsymbol{\phi}_{\tau_p}]$ contains $\tau_p$ pilots indexed by $\mathcal{T}=\{1,2,\dots,\tau_p\}$. The pilots are: \emph{a) normalized} $\lVert\boldsymbol{\phi}_t\rVert_2^2=\tau_p,$ $\forall t\in\mathcal{T}$ and \emph{b) mutually orthogonal} $\boldsymbol{\phi}^{\htransp}_t\boldsymbol{\phi}_{t'}={0}, \ \forall t\neq t'$, such that $t,t'\in\mathcal{T}$. The pilot pool is shared by all $L$ APs and $|\mathcal{U}|$ inactive UEs with $\tau_p\ll|\mathcal{U}|$. Because of this, \emph{pilot collisions} happen whenever two or more inactive UEs choose simultaneously the same pilot to request access to the network.

\begin{figure}[!htbp] 
    \centering
    \includegraphics[trim=0mm 0 0mm 0, clip=true, width=.3\textwidth]{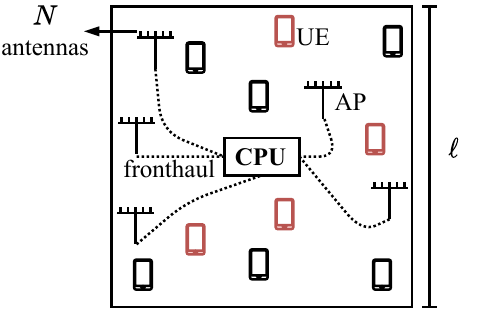}
    \caption{\small{Illustration of the CF-mMIMO system. Red UEs represent the portion $\mathcal{K}\subset\mathcal{U}$ of inactive UEs that wish to connect,} while black UEs remain idle. Fronthaul links are assumed to have unlimited capacity and provide error-free communication. {The proposed RA protocols depend on the information exchange between APs. The assumption of perfect cooperation between APs is valid for reliable wired links; otherwise, some performance degradation is expected.}}
    \label{fig:system-model}
\end{figure}

\subsection{Channel Model}
For $k\in\mathcal{K}$ and $l\in\mathcal{L}$, let $\mathbf{h}_{kl}\in\mathbb{C}^{N}$ denote the channel vector between the $k$-th UE and the $l$-th AP. {For tractability reasons}, we assume uncorrelated Rayleigh fading channels:
$$\mathbf{h}_{kl}\sim\mathcal{N}_{\mathbb{C}}(\mathbf{0},\beta_{kl}\mathbf{I}_N), \ \forall (k,l), \ k\in\mathcal{K}, \ l\in\mathcal{L}.$$
{This assumption considers a challenging environment with rich scattering due to movements of UEs and other objects. The results shown in the paper can easily be extended for the correlated Rayleigh fading model used in \cite{Demir2021}.}
The \emph{average channel gain} $\beta_{kl}$ is modeled according to
\begin{equation}
    \beta_{kl}\in\mathbb{R}_{+}=\Omega \cdot d^{-\zeta}_{kl},
    \label{eq:average-channel-gain}
\end{equation}
where $\Omega$ is a positive multiplicative power constant, $d_{kl}$ is the distance in meters between the $k$-th UE and the $l$-th AP, and $\zeta$ is the pathloss exponent.

\section{CF-SUCRe Protocol}\label{sec:cell-free-sucre}
The Ce-SUCRe protocol \cite{Bjornson2017} is comprised of {one preliminary step and} four main steps. {In this section,} we adapt these steps to a CF-mMIMO network {based on the user-centric perspective \cite{Demir2021}}, originating the CF-SUCRe protocol.
\vspace{-1mm}
\subsection{Step 0: Estimating Average Channel Gains}
In this preliminary step, the $k$-th UE estimates its set of average channel gains $\{\beta_{k1},\beta_{k2},\dots,\beta_{kL}\}$ relying on control beacons sent by the $L$ APs for $k\in\mathcal{U}$. For this to be possible, we assume that the control signaling of the APs are orthogonal to each other. {Below, we define the set of APs that a UE is capable of estimating.}
\vspace{-1mm}
{
\begin{definition}
    (Set of nearby APs: Influence region) For $k\in\mathcal{U}$, the $k$-th UE is only capable of perfectly estimating a channel gain $\beta_{kl}$ if and only if (iff) $q_l\beta_{kl}>\iota\cdot\sigma^2$ for $l\in\mathcal{L}$, where $q_l$ is the downlink (DL) transmit power of the $l$-th AP, $\iota\in\mathbb{R}_{+}$ is a multiplicative constant such that $\iota\geq1$, and $\sigma^2$ is the noise power. Then, we denote as $\mathcal{C}_k\subset\mathcal{L}$ the \textbf{set of nearby APs}, which is defined by $\mathcal{C}_k=\{l : q_l\beta_{kl}>\iota\cdot\sigma^2, \ \forall l\in\mathcal{L}\}$. {Moreover, we denote as $\check{\mathcal{C}}_k$ the so-called \textbf{natural set of nearby APs} obtained when $\iota=1$, which represents the largest number of APs that the $k$-th UE can know.}
\end{definition}
}

{
The definition above is based on the assumption that a UE is capable of detecting and perfectly estimating an average channel gain $\beta_{kl}$ with a power at least greater than the noise power $\sigma^2$, which is reasonable but also optimistic assumption. Therefore, the parameter $\iota$ is introduced to make this assumption stricter and more realistic. Note that UEs also need to estimate $\sigma^2$.
}

The physical interpretation of the set $\mathcal{C}_k$ is that it represents the APs located close to the $k$-th UE. To see this, by using \eqref{eq:average-channel-gain}, we can write the inequality $q_l\beta_{kl}>\iota\cdot\sigma^2$ as a function of the distance $d_{kl}$, yielding in
\begin{equation}
    d_{kl}<\left(\dfrac{1}{\iota}\cdot\dfrac{\Omega \cdot q_l}{\sigma^2}\right)^{\frac{1}{\zeta}}.
\end{equation}
This means that $\mathcal{C}_k$ can be geometrically interpreted as the set of APs whose APs have distances relative to the $k$-th UE less than the \emph{limit distance} or \emph{limit radius} given by
\begin{equation}
    d^{\text{lim}}=d^{\text{lim}}_{k}=\left(\dfrac{1}{\iota}\cdot\dfrac{\Omega \cdot q_l}{\sigma^2}\right)^{\frac{1}{\zeta}}.
    \label{eq:limit-distance}
\end{equation}
Note that we can drop the subscript $k$ of the limit radius. Further, the limit radius is as large as possible when $\iota=1$. Consequently, the larger the $\iota$, the smaller are $d^{\text{lim}}$ and $|\mathcal{C}_k|$. In the worst case, we assume that $|\mathcal{C}_k|\geq1$ APs irrespective of the value of $\iota$, meaning that at least one AP has to be known so that an inactive UE can try to access the network. Particularly, the circular region around a UE with radius $d^{\text{lim}}$ can be interpreted as the region which contains the APs that most influence a UE's communication, namely, the \emph{UE's influence region} \cite{Demir2021}. 

\subsection{Step 1: Pilot Transmission and Pilot Activity}
At a given instant, each inactive UE from $\mathcal{K}$ selects a pilot $\boldsymbol{\phi}_t\in\mathbb{C}^{\tau_p}$ at random from the pilot pool $\boldsymbol{\Phi}$, which is denoted as $c(k)\in\mathcal{T}$ for $k\in\mathcal{K}$. After the $\lvert\mathcal{K}\rvert$ UEs transmit their chosen pilots in a broadcast form, the $l$-th AP receives:
\begin{equation}
    \mathbf{Y}_{l}\in\mathbb{C}^{N\times\tau_p} = \sum_{k\in\mathcal{K}} \sqrt{p_k} \mathbf{h}_{kl} \boldsymbol{\phi}^{\transp}_{c(k)} + \mathbf{N}_l,
    \label{eq:pilot-tx}
\end{equation}
where $p_k$ is the UL transmit power of the $k$-th UE and $\mathbf{N}_l\in\mathbb{C}^{N\times\tau_p}$ is the receiver noise matrix with i.i.d. elements distributed as $\mathcal{N}_{\mathbb{C}}({0},\sigma^2)$. Then, the $l$-th AP correlates its received signal with each pilot available in the pilot pool $\boldsymbol{\Phi}$. The correlation with the $t$-th pilot yields in
\begin{equation}
    \mathbf{y}_{lt}\in\mathbb{C}^{N}=\mathbf{Y}_{l} \dfrac{\boldsymbol{\phi}_t^{*}}{\lVert\boldsymbol{\phi}_t\rVert_2}=\sum_{i\in\mathcal{S}_t} \sqrt{p_i\tau_p}\mathbf{h}_{il}+\mathbf{n}_{lt},
    \label{eq:UL-correlated-signal}
\end{equation}
where $\mathbf{n}_{lt}\in\mathbb{C}^{N}\sim\mathcal{N}_{\mathbb{C}}(\mathbf{0},\sigma^2\mathbf{I}_{N})$ is the effective receiver noise vector. The {\emph{set of colliding UEs} is denoted as} $\mathcal{S}_t\subset\mathcal{K}$ {for} $t\in\mathcal{T}$. Hence, we have that $\lvert\mathcal{S}_t\rvert>1$ UEs if a collision occurs.

\subsubsection{{Pilot activity}} On one hand, APs are unable to resolve collisions using the correlated signal $\mathbf{y}_{lt}$, since they do not have any information about the $|\mathcal{K}|$ UEs. On the other hand, a utility of $\mathbf{y}_{lt}$ is to detect which of the pilots are being used or are active \cite{Bjornson2017}. Locally, the $l$-th AP can obtain:
{
\begin{align}
    \dfrac{1}{N}\lVert\mathbf{y}_{lt}\rVert^2_2&=\frac{1}{N}\sum_{n'=1}^{N}\Bigg[\left(\sum_{i\in\mathcal{S}_t} \sqrt{p_i\tau_p}\lvert{h}^{(n')}_{il}\rvert\right)^2+2\left(\sum_{i\in\mathcal{S}_t} \sqrt{p_i\tau_p}\lvert{h}^{(n')}_{il}\rvert\right)\nonumber\\
    &\cdot\left(\lvert n^{n'}_{lt} \rvert\right) + \left(\lvert n^{n'}_{lt} \rvert\right)^2\Bigg]
    \,\,   \xrightarrow{N\rightarrow\infty} \,\, \underbrace{\sum_{i\in\mathcal{S}_t}p_i\tau_p\beta_{il}}_{\alpha_{lt}} \, + \, \sigma^2,
    \label{eq:check-pilot-t}
\end{align}
where the approximation relies on the law of large numbers and the facts that: channels between different UEs are uncorrelated and noise is independent.\footnote{We assume that the readers are familiar with channel hardening and favorable propagation concepts. For a more comprehensive definition, we point out the interested reader to \cite{Demir2021}.} We let $\alpha_{lt}=\sum_{i\in\mathcal{S}_t}p_i\tau_p\beta_{il}$ denote the \emph{UL signal power of colliding UEs} for $l\in\mathcal{L}$. Moreover, we say that the $t$-th pilot is an \emph{active pilot} at the $l$-th AP if $\frac{1}{N}\lVert\mathbf{y}_{lt}\rVert^2_2>\sigma^2$.} Altogether, the $l$-th AP knows
\begin{equation}
    \tilde{\mathbf{a}}_l\in\mathbb{R}_{+}^{\tau_p}=\dfrac{1}{N}\left[{\lVert\mathbf{y}_{l1}\rVert^2_2},{\lVert\mathbf{y}_{l2}\rVert^2_2},\dots,{\lVert\mathbf{y}_{l\tau_p}\rVert^2_2}\right]^\transp,
\end{equation}
where $\tilde{\mathbf{a}}_l$ measures the pilot activity at the $l$-th AP for $l\in\mathcal{L}$. Importantly, we can use $\tilde{\mathbf{a}}_l$ to design a user-centric {solution for the CF RA.} This can be done by restricting the amount of pilots each AP {serves} based on the relative powers stored in $\tilde{\mathbf{a}}_l$. To do this, {we consider that} all $L$ APs send their $\tilde{\mathbf{a}}_l$'s to the CPU. Since $\tau_p$ is constant in the order of tens, the communication of $\tau_p$-length vectors $\tilde{\mathbf{a}}_l$ through fronthaul links is computationally scalable.\footnote{{We adopt the definition of computational scability from \cite{Demir2021}.}} Then, the CPU knows
\begin{equation}
    \tilde{\mathbf{A}}\in\mathbb{R}_{+}^{\tau_p\times L}=[\tilde{\mathbf{a}}_1, \tilde{\mathbf{a}}_2, \dots, \tilde{\mathbf{a}}_L].
    \label{eq:matrix-A-tilde}
\end{equation}
The CPU can now define which APs are the most suitable to serve the $t$-th pilot based on $\tilde{\mathbf{A}}$. To do so, we suggest the following heuristic procedure: \textbf{1)} Eliminate all \emph{irrelevant} entries of $\tilde{\mathbf{A}}$: set all entries in which $\frac{1}{N}\lVert\mathbf{y}_{lt}\rVert^2_2\leq\sigma^2$ to zero, $\forall l\in\mathcal{L}, \ \forall t\in\mathcal{T}$; \textbf{2)} Define $L^{\max}$ as the \emph{number of APs that serves each pilot}. For simplicity, the integer $1 \leq L^{\max} \leq L$ is {reasonably} assumed to be the same for every pilot $t\in\mathcal{T}$, {since pilots are equally likely to be chosen by a UE.}\footnote{{A more adequate strategy would be to allocate more APs to the highest valued $\tilde{\mathbf{A}}$ inputs, giving more emphasis to pilots which have more collisions. However, we restricted ourselves to the simplest case for the sake of mathematical tractability. Future works may collaborate with better ways to select the more general parameter $L^{\max}_t$ by better studying properties of $\tilde{\mathbf{A}}$.}} Then, we define the following set of APs.

{
    \begin{definition}
    (Set of pilot-serving APs) For each $t$-th row of $\tilde{\mathbf{A}}$, the indices of the $L^{\max}$ highest (non-zero) entries leads to a set $\mathcal{P}_t\subset\mathcal{L}, \,\forall t\in\mathcal{T}$. The \textbf{set of pilot-serving APs}, $\mathcal{P}_t$, indicates the ablest APs chosen to serve the $t$-th pilot, where $|\mathcal{P}_t|\leq L^{\max}$ and the pilot is {inactive} for $|\mathcal{P}_t|=0$ or $\mathcal{P}_t=\emptyset$.
    \end{definition}
}

{After defining the $\mathcal{P}_t$'s,} the CPU {informs} the $L$ APs of {its decision}. For $l\in\mathcal{L}$, we let $\mathcal{T}_l\subset\mathcal{T}$ denote the set of pilots served by the $l$-th AP{, whose construction is} $\mathcal{T}_l=\{t : l\in\mathcal{P}_t, \ \forall t \in\mathcal{T}\}$. {Note that} the $l$-th AP can still serve more than one pilot or $0\leq|\mathcal{T}_l|\leq\tau_p$, where $|\mathcal{T}_l|=0$ denotes the \emph{inoperative} case where the AP does not participate in the RA phase. {The main benefit of limiting the number $L^{\max}$ of APs that are serving each pilot is to mitigate the computational complexity per AP and the inter-AP interference that could arise from multiple non-coherent transmissions\footnote{{Next, we will assume that APs transmit only using local knowledge via the maximum-ratio scheme. Another approach would be to use more robust DL transmission methods, but at the cost of increasing computational complexity.}} of APs. Thus, the CF-SUCRe protocol can become more energy efficient, since not all $L$ APs must necessarily operate in the RA phase.}

The physical meaning of $\mathcal{P}_t$ is that it contains the indices of the $L^{\max}$ APs that are closer to the colliding UEs in $\mathcal{S}_t$. To see this, as before, we seek for a geometric interpretation of $\mathcal{P}_t$ with respect to the distance $d_{il}$, for $i\in\mathcal{S}_t$ and $l\in\mathcal{L}$. From \eqref{eq:average-channel-gain}, \eqref{eq:check-pilot-t}, and the inequality $\frac{1}{N}\lVert\mathbf{y}_{lt}\rVert^2_2>\sigma^2$, the set $\mathcal{P}_t$ contains the AP indices of the last $L^{\max}$ vector entries:
\begin{equation}
    \mathrm{argsort}\left(\left[\sum_{i\in\mathcal{S}_t}p_i d^{-\zeta}_{i1},\,\, \sum_{i\in\mathcal{S}_t}p_i d^{-\zeta}_{i2}, \,\,\dots\,\, , \,\,\sum_{i\in\mathcal{S}_t}p_i d^{-\zeta}_{iL}\right]^\transp\right),
    \label{eq:geometric-interpretation-pt}
\end{equation}
where the closer the $l$-th AP is to the UEs in $\mathcal{S}_t$, the greater the sum $\sum_{i\in\mathcal{S}_t}p_i d^{-\zeta}_{il}$ for $l\in\mathcal{L}$. {The above expression reveals that the UL transmit power $p_i$ can also change considerably the construction of $\mathcal{P}_t$'s.} But, if all UEs transmit with equal power $p_1=\dots=p_{|\mathcal{S}_t|}=p$, the dependency on $p$ is gone. In this special case, the construction of $\mathcal{P}_t$ is only controlled by the parameter $L^{\max}$, since the other variables are defined by the geometry and the environment of a scenario of interest.

\subsubsection{{Obtaining sets of nearby and pilot-serving APs}} Fig. \ref{fig:illustration-ccal-pcal} illustrates the process of obtaining sets $\check{\mathcal{C}}_k$ and $\mathcal{P}_t$ for $L^{\max}=10$ APs. For illustrative reasons, we consider that the asymptotic approximation in \eqref{eq:check-pilot-t} is obtained. There are $\tau_p=3$ active pilots, where for each pilot a collision of size $|\mathcal{S}_t|=2$ UEs occurs. This gives rise to the colliding sets: $\mathcal{S}_1$ '$\blacksquare$', $\mathcal{S}_2$ '$\blacktriangleleft$', and $\mathcal{S}_3$ '$\blacktriangleright$'. These markers denote the positions of the UEs. The {natural-nearby} APs that comprise $\check{\mathcal{C}}_k$ of each UE $k\in\mathcal{S}_t$ are the ones within the UE's influence region delimited by the colored circles of radius $d^{\text{lim}}$, which is calculated according to \eqref{eq:limit-distance}. The {sets} of {pilot-serving} APs $\mathcal{P}_1,\mathcal{P}_2,\mathcal{P}_3$ are also differentiated by markers.

\begin{figure}[!htbp]
    \centering
    \vspace{-1mm}
    \includegraphics[width=.4\textwidth]{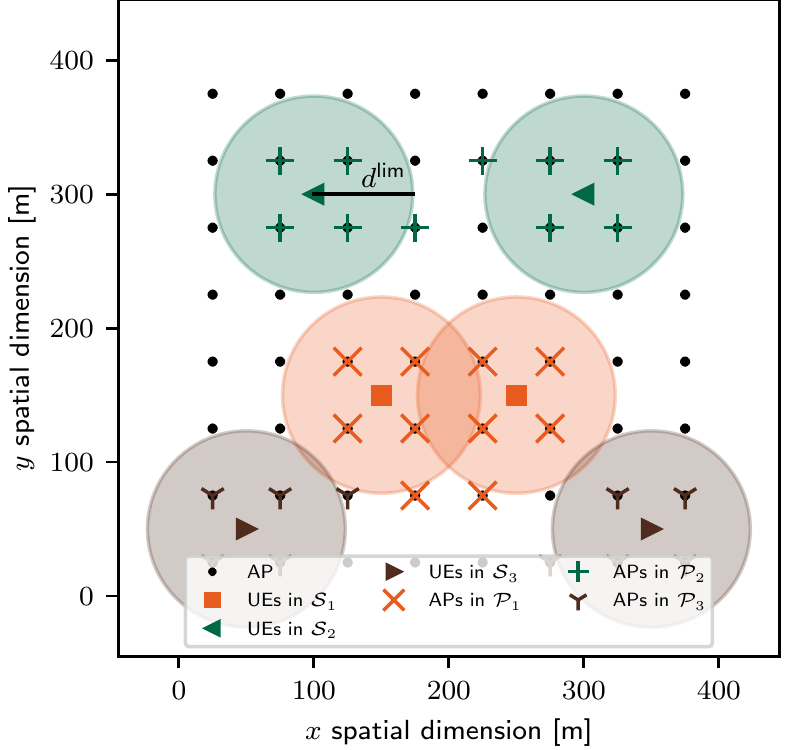}
    \vspace{-1mm}
    \caption{\small Construction of the sets $\check{\mathcal{C}}_k$ and $\mathcal{P}_t$ for $L^{\max}=10$ APs. A total of $L=64$ APs are disposed in an 8 $\times$ 8 square grid for $\ell=400$ m. Each color represents a pilot $t\in\mathcal{T}=\{1,2,3\}$. Each pilot has $|\mathcal{S}_t|=2$ UEs colliding in different spatial arrangements: $\mathcal{S}_1$ '$\blacksquare$', $\mathcal{S}_2$ '$\blacktriangleleft$', and $\mathcal{S}_3$ '$\blacktriangleright$'. The colored circles with radius $d^{\text{lim}}\approx73.30$ m mark the UE's influence region. The APs '$\bullet$' within of the colored circles are in $\check{\mathcal{C}}_k$ for $k\in\mathcal{S}_t$. APs with different markers superimposed over '$\bullet$' indicates the construction of $\mathcal{P}_1$, $\mathcal{P}_2$, and $\mathcal{P}_3$. Fixed parameters are: $\Omega=-30.5$ dB, $\zeta=3.67$, $\sigma^2=-94$ dBm, $\tau_p=3$ pilots, $p=100$ mW, and $q_l={\frac{200}{64}}$ mW.}
    \label{fig:illustration-ccal-pcal}
    \vspace{-1mm}
\end{figure}

{From Fig. \ref{fig:illustration-ccal-pcal}, one can observe that $\mathcal{P}_t$ almost match $\cup_{i\in\mathcal{S}_t}\check{\mathcal{C}}_i$ for a proper choice of $L^{\max}$. As a result, we state the following remark.
}

{
\begin{remark}\label{remark:similarity}
    (Similarity between set of nearby APs, $\mathcal{C}_k$, and set of pilot-serving APs, $\mathcal{P}_t$) For $i\in\mathcal{S}_t$, due to the {different construction numerologies}, the sets $\mathcal{C}_i$ (nearby APs) and $\mathcal{P}_t$ (pilot-serving APs) can be very different. It is reasonable to expect, however, that $\mathcal{C}_i\cap\mathcal{P}_t\neq\emptyset$ for most of the colliding UEs. The intuition behind this expectation is because $\mathcal{P}_t$ depends on $\sum_{i\in\mathcal{S}_t}p_i\tau_p\beta_{il}$, while $\mathcal{C}_i$ on the individual $q_l\beta_{il}$'s. Hence, both quantities are positively correlated to the magnitude of average channel gains $\beta_{il}$ for $i\in\mathcal{S}_t$ and $l\in\mathcal{L}$. Both sets are equal in the very special case where $|\mathcal{S}_t|=1$, $p=q_l$, and a proper choice of $L^{\max}$.
\end{remark}
}

\subsection{Step 2: Precoded RA Response}
Let $\mathcal{P}=\cup_{t=1}^{\tau_p}\mathcal{P}_t$ denote all the APs that serve at least one pilot. {For $l\in\mathcal{P}$, the $l$-th} AP sends a precoded DL pilot signal employing a multi-cast maximum-ratio (MR) transmission:
\begin{equation}
    \mathbf{V}_l\in\mathbb{C}^{N\times\tau_p}=\sqrt{q_l}\sum_{t\in\mathcal{T}_l} \dfrac{\mathbf{y}_{lt}}{\lVert\mathbf{y}_{lt}\rVert_2}\boldsymbol{\phi}^{\transp}_t.
    \label{eq:DL-precoding}
\end{equation}
Note that the MR precoding vector ${\mathbf{y}_{lt}}/{\lVert\mathbf{y}_{lt}\rVert}_2$ spatially directs the $t$-th pilot towards the colliding UEs in $\mathcal{S}_t$. For $k\in\mathcal{S}_t$, the $k$-th UE receives
\begin{equation}
    \mathbf{z}^{\transp}_{k}\in\mathbb{C}^{\tau_p}=\sum_{l\in\mathcal{P}}\mathbf{h}^{\htransp}_{kl}\mathbf{V}_l+\boldsymbol{\eta}^{\transp}_k,
\end{equation}
where $\boldsymbol{\eta}_k\sim\mathcal{N}_{\mathbb{C}}(\mathbf{0},\sigma^2\mathbf{I}_{\tau_p})$ is the receiver noise. Then, the $k$-th UE correlates $\mathbf{z}^{\transp}_{k}$ with the pilot $\boldsymbol{\phi}_t$ used in Step 1, yielding in:
\begin{align}
    z_k \in \mathbb{C} & = \mathbf{z}^{\transp}_{k} \dfrac{\boldsymbol{\phi}_t^{*}}{\lVert\boldsymbol{\phi}_t\rVert_2} = \sum_{l\in\mathcal{P}_t} \sqrt{q_l\tau_p}\mathbf{h}^{\htransp}_{kl}\dfrac{\mathbf{y}_{lt}}{\lVert\mathbf{y}_{lt}\rVert_2}+\eta_{kt},\nonumber\\
    & \stackrel{(a)}{=} \underbrace{\vphantom{\left(\sum_{l\in\mathcal{P}_t} \sqrt{\dfrac{q_lp_i\tau^2_p}{\lVert\mathbf{y}_{lt}\rVert^2_2}}{\mathbf{h}_{kl}^\htransp\mathbf{h}_{il}}\right)}\sum_{l\in\mathcal{P}_t} \sqrt{\dfrac{q_lp_k\tau^2_p}{\lVert\mathbf{y}_{lt}\rVert^2_2}}{\lVert\mathbf{h}_{kl}\rVert^2_2}}_{\text{effective channel}}+\sum_{i\in\mathcal{S}_t\setminus\{k\}}\underbrace{\left(\sum_{l\in\mathcal{P}_t} \sqrt{\dfrac{q_lp_i\tau^2_p}{\lVert\mathbf{y}_{lt}\rVert^2_2}}{\mathbf{h}_{kl}^\htransp\mathbf{h}_{il}}\right)}_{\text{effective interfering channel}} \nonumber \\
    & +\underbrace{\sum_{l\in\mathcal{P}_t}\sqrt{\dfrac{q_l\tau_p}{\lVert\mathbf{y}_{lt}\rVert^2_2}}{\mathbf{h}_{kl}^\htransp\mathbf{n}_{lt}}+\eta_{kt}}_{\text{noise}},
    \label{eq:DL-correlated-signal}
\end{align}
where in $(a)$ we used \eqref{eq:UL-correlated-signal} and ${\eta}_{kt}\sim\mathcal{N}_{\mathbb{C}}({0},\sigma^2)$ is the effective receiver noise. From the channel hardening definition {\cite[Eq. (2.23)]{Demir2021}}, the effective channel above satisfies
\begin{align}
     \dfrac{1}{\sqrt{N}}\sum_{l\in\mathcal{P}_t} \sqrt{\dfrac{q_lp_k\tau^2_p}{\lVert\mathbf{y}_{lt}\rVert^2_2}}{\lVert\mathbf{h}_{kl}\rVert^2_2} - \nonumber \\
     -\dfrac{1}{\sqrt{N}}\left(\mathbb{E}\left\{\sum_{l\in\mathcal{P}_t}\sqrt{\dfrac{q_lp_k\tau^2_p}{\lVert\mathbf{y}_{lt}\rVert^2_2}}{\lVert\mathbf{h}_{kl}\rVert^2_2}\right\}\right) & \xrightarrow[]{}0, \text{ as } N \xrightarrow[]{}\infty.
    \label{eq:asymptotic-z}   
\end{align}
{This means that the difference between the instantaneous effective channel and its mean value converges strongly to zero when the number of antennas per AP, $N$, is very large. Note that the asymptotic analysis in \eqref{eq:check-pilot-t} was considered for $\mathbf{y}_{lt}$. Hence, if favorable propagation (asymptotically orthogonality of interfering channels) also holds, we can approximate \eqref{eq:DL-correlated-signal} using only the effective parts; the noise part also converges to zero for large $N$ since channels and noise are uncorrelated. By using the almost sure convergence of \eqref{eq:asymptotic-z} and solving the expectation, we get for the effective channel term in \eqref{eq:DL-correlated-signal}:}
\begin{equation}
    \dfrac{\Re(z_k)}{\sqrt{N}}\approx \tilde{z}_k = \sum_{l\in\mathcal{P}_t}\left(\dfrac{\sqrt{q_lp_k}\tau_p\beta_{kl}}{\sqrt{\alpha_{lt}+\sigma^2}}\right),
    \label{eq:approx}
\end{equation}
where $\Re(z_k)$ stands for the real part of $z_k$ and $\alpha_{lt}$ was defined in \eqref{eq:check-pilot-t}. The {approximation} $\tilde{z}_k$ in \eqref{eq:approx} is handy for the {$k$-th} UE to provide a way to compare its \emph{own total UL signal power}:
\begin{equation}
    \gamma_k\in\mathbb{R}_{+}=\sum_{l'\in\mathcal{C}_k}p_k\tau_p\beta_{kl'}
    \label{eq:gamma_k}
\end{equation}
with the \emph{total UL signal power of colliding UEs} in $\mathcal{S}_t$:
\begin{equation}
    \alpha_{t}\in\mathbb{R}_{+}=\sum_{l\in\mathcal{P}_t}\alpha_{lt}=\sum_{i\in\mathcal{S}_t}\left(\sum_{l\in\mathcal{P}_t}p_i\tau_p\beta_{il}\right),
    \label{eq:alpha_t}
\end{equation}
where we used the definition of $\alpha_{lt}$ from \eqref{eq:check-pilot-t}. {The $k$-th} UE knows $\gamma_k$, whereas $\alpha_t$ is unknown to it. The key idea behind SUCRe introduced in \cite{Bjornson2017} is to resolve a collision by having the strongest UE among those colliding re-transmitting the $t$-th pilot. In order for the $k$-th UE to be able to determine if it is the strongest UE, it needs to estimate $\alpha_t$ having the correlated received signal $z_k$ in \eqref{eq:DL-correlated-signal} as information. {Then, the $k$-th UE can make a decision based on comparing $\alpha_{lt}$ with $\gamma_k$.} We will denote as $\hat{\alpha}_{t,k}$ \emph{any estimation of $\alpha_t$} made by the $k$-th UE. In Section \ref{sec:estimators}, we discuss different estimators for $\hat{\alpha}_{t,k}$.

\subsection{Step 3: Contention Resolution \& Pilot Repetition}
In this step, the pilot collisions are solved through a distributed process called as \emph{contention resolution}. Based on SUCRe \cite{Bjornson2017} and the asymptotic analysis above, each UE applies the rule below to decide if it is the contention winner {\cite{Bjornson2017}}\footnote{{Different from \cite{Bjornson2017} {and for tractability purpose in the new CF system context}, we consider no bias parameters in the decision rule.}}:
\begin{equation}
    \begin{aligned}
        {R}_k: & \ \gamma_k > \dfrac{{\hat{\alpha}}_{t,k}}{2} \ \text{(repeat)},\\
        {I}_k: & \ \gamma_k \leq \dfrac{{\hat{\alpha}}_{t,k}}{2} \ \text{(inactive)},
    \end{aligned} 
    \label{eq:decision}
\end{equation}
{for $k\in\mathcal{K}$ and} where ${\hat{\alpha}}_{t,k}$ denotes any estimation of $\alpha_t$ made by {the $k$-th} UE {which selected the $t$-th pilot, for $t\in\mathcal{T}$}. This \emph{distributed decision rule} reads as follows: {the $k$}-th UE re-transmits pilot signal $\boldsymbol{\phi}_t$ {if it claims itself as the strongest, where the notion of strength is related to $\frac{{\hat{\alpha}}_{t,k}}{2}$ \cite{Bjornson2017},} and hypothesis $R_k$ is therefore true; otherwise, {the $k$-th} UE concludes that hypothesis $I_k$ is true, {deciding} to pull out of and postpone the access attempt. Step 3 finishes with the re-transmission of the same pilots sent in Step 1 by the UEs that have decided for $R_k$. {We denote as $\mathcal{W}_t\subseteq\mathcal{S}_t$ the \emph{set of winning UEs} that decided to re-transmit the $t$-th pilot.} {The transmission in Step 3 also contains the identity of the UE and a request for payload transmission, resembling the connection request in legacy protocols \cite{Popovski2020,Bjornson2017}.}

\subsection{Step 4: Allocation of Dedicated Data Payload Pilots}
The $L$ APs receive the pilots repeated by the winning contention UEs {in $\mathcal{W}_t,\,\forall t\in\mathcal{T}$}. {Clearly,} the $l$-th AP only needs to check the re-transmitted pilots that are served by it, which are specified by $\mathcal{T}_l\subset\mathcal{T}$ for $l\in\mathcal{L}$. {In \cite{Bjornson2017}, a winning UE only successfully accesses the network if it retransmits the pilot alone, meaning that $|\mathcal{W}_t|=1$ for an access to be considered successful in the Ce-SUCRe. Herein, we demonstrate that it is possible to solve collisions even if $|\mathcal{W}_t|>1$ through the concept of \emph{spatial separability}, which rises from a spatial-based reuse of pilots and the adoption of a user-centric perspective (related to the definition of $\mathcal{C}_k$ and $\mathcal{P}_t$) {to design the network}. This concept is presented in more depth in the next section. If the access of a UE is considered successful, the APs that serve it can estimate its channel and successfully decode the messages used to finally establish its network connection.}

\section{{Spatial Separability}}\label{sec:spatial-separability}
{In this section, we present the concept of \emph{spatial separability}. Spatial separability resolves the access contention even if multiple UEs claim themselves as contentions winners after Step 3 through a spatial-based reuse of pilots and the adoption of the user-centric design perspective. We define it as follows.} 

{
\begin{definition}\label{def:3}
(Spatial separability) {For the $t$-th pilot, consider the $k$-th winning UE for $k\in\mathcal{W}_t$ and $t\in\mathcal{T}$. Let $\underline{\check{\mathcal{C}}_k}=\cup_{i\in\mathcal{W}_t\setminus\{k\}}\check{\mathcal{C}}_i$ denote the set of nearby APs closest to the other winning UEs. The $k$-th UE is \textbf{spatially separable} if the following conditions are simultaneously satisfied:
\begin{align}
\textbf{(a)}&&\mathcal{P}_t\cap\check{\mathcal{C}}_k \neq \emptyset\label{eq:spatial-separability:cond1}\\
\textbf{(b)}&&(\mathcal{P}_t\cap\check{\mathcal{C}}_k)\setminus(\mathcal{P}_t\cap\underline{\check{\mathcal{C}}_k}) \neq \emptyset\label{eq:spatial-separability:cond2}
\end{align}
where \textbf{(a)} means that at least one pilot-serving AP is within the influence region of the $k$-th UE and \textbf{(b)} indicates that there is at least one pilot-serving AP that is \textbf{exclusively} within the influence region of the $k$-th UE. In this way, the pilot re-transmitted by the $k$-th UE in Step 3 together with a connection request message will be successfully decoded by the APs that are exclusively closest to it, guaranteeing its access.}
\end{definition}
}

{
Condition \textbf{\textit{(b)}} in \eqref{eq:spatial-separability:cond2} is reasonable because the effective power of a winning UE $i\in\mathcal{W}_t\setminus\{k\}$ at the border of its influence region will decay on average almost $\beta^{\lim}\approx-100$ dB for: $\Omega=-30.5$ dB, $\zeta=3.67$, and $d^{\lim}=73.30$ m.\footnote{{A better evaluation of condition \textbf{\textit{(b)}} is let for future work, since a more refined definition of it may consider other power-related metrics.}} We then give the following example.  
} 

\vspace{2mm}
\noindent\textbf{Example}. \emph{{(Two winning UEs)} Let us assume the case of two winning UEs in $\mathcal{W}_t=\{1,2\}$. Given a suitable choice of $L^{\max}$, the winning UEs are served by the set $\mathcal{P}_t$ of pilot-serving APs. By checking $(\mathcal{P}_t\cap\check{\mathcal{C}}_1)\setminus(\mathcal{P}_t\cap\check{\mathcal{C}}_2)$ and vice-versa, it is possible to see if some of the APs in $\mathcal{P}_t$ are only close to one of the two colliding UEs. If that is the case, it is reasonable to assume that those APs that are only close to one of the UEs are still able to spatially separate that UE, seeing that the interference from the re-transmission of the other UE is low due to increased distance. Fig. \ref{fig:separability} illustrates the {discussed \textbf{spatial separability} concept.} Hence, when two colliding UEs declare themselves winners in Step 3, there can be: the acceptance of both, the acceptance of one of the two, or the acceptance of neither, depending on the subsets $\check{\mathcal{C}}_i$ and $\mathcal{P}_t$ for $i\in\mathcal{W}_t$.}

\begin{figure}[htp]
    \centering
    \vspace{-3mm}
    \includegraphics[width=.50\textwidth]{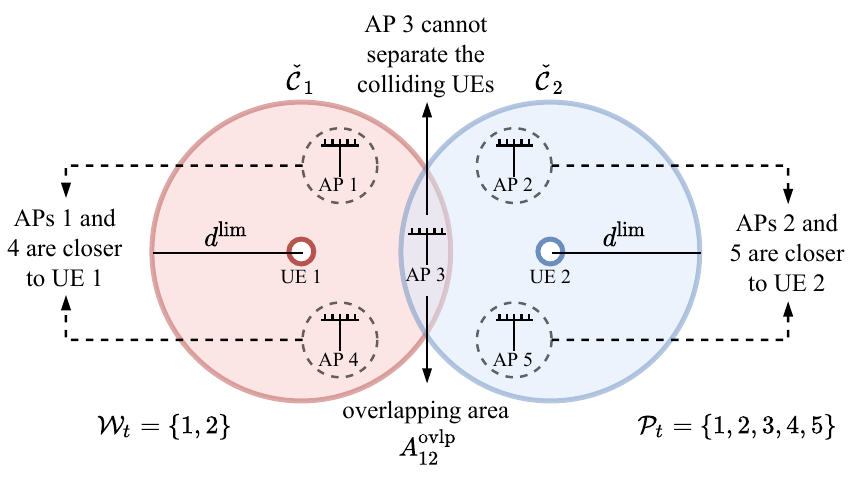}
    \vspace{-3mm}
    \caption{\small Illustration of the concept of spatial separability when considering a set of two colliding UEs {$\mathcal{W}_t=\{1,2\}$} served by five APs in $\mathcal{P}_t=\{1,2,3,4,5\}$.}
    \label{fig:separability}
\end{figure}

\subsection{{Analyzing the Spatial Separability}}\label{subsec:analysis}
The potential of spatial separability is analyzed through the evaluation of the veracity of conditions \textit{\textbf{(a)}} and \textit{\textbf{(b)}} in \eqref{eq:spatial-separability:cond1} and \eqref{eq:spatial-separability:cond2}, respectively. In this part, we assume that $\mathcal{W}_t=\mathcal{S}_t$ so that the analysis is independent of both the estimate $\hat{\alpha}_{t,k}$ and decision in \eqref{eq:decision} and is, consequently, more generalist.

We denote by $(x_k,y_k)$ the pair of coordinates that describes the position of the $k$-th UE, for $k\in\mathcal{S}_t$. We assume that $x_k,y_k\sim \rm U_{[0,\ell]}$, where $x_k$ and $y_k$ are independent and $\ell$ is the square length (see Fig. \ref{fig:system-model}). The distance between two winning UEs is then $d_{ij}=((x_i-x_j)^2+(y_i-y_j)^2)^{1/2}$, for $i,j\in\mathcal{S}_t$. Thereby, the cumulative distribution function (CDF) of $d_{ij}$ is\footnote{{The steps to obtain the CDF of the distance involves: a) triangular distribution $x_i-x_j$, b) another triangular distribution $|x_i-x_j|$, c) evaluate distance $d_{ij}$ distribution considering the previous steps.}}
\begin{equation}
F_{d_{ij}}(d) = \frac{2}{\ell^4}
\begin{cases}
g(d)-g(0),&\text{for } 0\leq d< \ell \\
2g(\ell)-g\left(\sqrt{d^2-\ell^2}\right)-g(0),&\text{for } \ell\leq d\leq \sqrt{2}\ell
\end{cases},
\label{eq:cdf-distance}
\end{equation}
where function $g:\mathbb{R}_{+}\mapsto\mathbb{R}_{+}$ is defined as
\begin{align*}
    g(z)&=\dfrac{\ell}{3}\sqrt{d^2-z^2}(z(3\ell-2z)+2d^2)+\ell^2d^2\arctan\left(\dfrac{z}{\sqrt{d^2-z^2}}\right)\\
    &-\ell d^2 z +\dfrac{\ell z^3}{3}+\dfrac{\ell^2z^2}{2}-\dfrac{z^4}{4}.
\end{align*}

{Let $A_{k}^{\text{dom}}$ denote the \emph{dominant area} in which the signal of the $k$-th UE is dominant (\emph{e.g.,} the remaining area of UE 1 without the overlapping region in Fig. \ref{fig:separability}). We approximate the calculation of such area through the following expression:
\begin{equation}
A_{k}^{\text{dom}}= A_{k} - \mathbb{E}\left\{\sum_{i\in\mathcal{S}_t, i \neq k}F_{d_{ki}}(2d^{\lim})A^{\text{ovlp}}_{ki}\right\},
\label{eq:adom}
\end{equation}
where $A_{k}=\pi(d^{\lim})^2$ 
is the influence region of the $k$-th UE given the limit radius defined in \eqref{eq:limit-distance}. The overlap between the influence region of two UEs only occurs if the distance between the UEs is less or equal than two limit radius $2d^{\lim}$ (see Fig. \ref{fig:illustration-ccal-pcal}). We call two winning UEs {\emph{correspondents}} if their areas overlap. Therefore, the probability of two UEs being {correspondents} is $F_{d_{ki}}(2d^{\lim})=\mathbb{P}\{d_{ki}\leq2d^{\lim}\}$. Finally, the area $A^{\text{ovlp}}_{ki}$ stands for the \emph{overlapping area} between two circular influence regions of same radius $d^{\lim}$.
}

{
The expression in \eqref{eq:adom} only considers first-order overlaps, meaning that overlaps are disjoint, that is, we have not considered overlaps between overlaps. This is reasonable based on the fact that the probability of two UEs requesting access being {correspondents} \emph{and} simultaneously choosing the same pilot results in $\frac{1}{\tau_p} F_{d_{ki}}(2d^{\lim})\approx0.024$ for $\tau_p=5$, $d^{\lim}\approx73.30$, and $\ell=400$. This becomes even less likely to occur if the effect of the decision criterion in \eqref{eq:decision} is considered. Thus, Eq. \eqref{eq:adom} can be further approximated as $A_{k}^{\text{dom}}= A_{k} - F_{d_{ki}}(2d^{\lim}){\max(\overline{\lvert{\mathcal{S}}_t\rvert}-1,\, 0)} \mathbb{E}\{A^{\text{ovlp}}_{ki}\}$, where $\overline{\lvert{\mathcal{S}}_t\rvert}=({\lvert \mathcal{U}\rvert P_a})/{\tau_p}$ is the average collision size. Our goal is to compute the \emph{expected overlapping area} $\bar{A}^{\text{ovlp}}=\mathbb{E}\{A^{\text{ovlp}}_{ki}\}$.
}

{
The overlapping area between two similar circles is:
\begin{equation}
    A_{ki}^{\text{ovlp}} = \underbrace{2(d^{\lim})^2\arccos\left(\frac{d_{ki}}{2d^{\lim}}\right)}_{h_1(d_{ki})} - \underbrace{\vphantom{\left(\sqrt{\frac{d_{ki}}{2d^{\lim}}}\right)}\frac{d_{ki}}{2}\sqrt{4(d^{\lim})^2 - d_{ki}^2}}_{h_2(d_{ki})},
\end{equation}
where $A^{\text{ovlp}}_{ki}$ is a random variable due to $d_{ki}$ (see Fig. \ref{fig:separability}). Since we are interested in its expected value, we divide the process of its obtaining in two parts: \textbf{i)} $\mathbb{E}\{h_1(d_{ki})\}$ and \textbf{ii)} $ \mathbb{E}\{h_2(d_{ki})\}$. In addition, it is worth noting the fact that $2d_{\lim}\ll\ell$ in practice, meaning that only the first case in \eqref{eq:cdf-distance} is relevant for us. Thus, the probability density function (PDF) of the distance $d_{ki}$ is:
\begin{equation}
    f_{d_{ij}}(d) = \frac{2}{\ell^4}((1+\pi)\ell^2d-4\ell d^2-d^3)\,\,\,\,\text{for } 0\leq d\leq \ell.
\end{equation}
By using the law of the unconsciousness statistician, 
we get: 
\begin{equation}
    \mathbb{E}\{h_1(d_{ki})\}=\int_ {0}^{\ell}{2(d^{\lim})^2\arccos\left({\frac{d}{2d^{\lim}}}\right)}f_{d}(d) \, {\rm d}d,
\end{equation}
\begin{equation}
    \mathbb{E}\{h_2(d_{ki})\}=\int_{0}^{\ell} \frac{d}{2}\sqrt{4(d^{\lim})^2 - d^2}f_{d}(d) \, {\rm d}d.
\end{equation}
Note that we constrain the range of the $\arccos(\cdot)$ function to be one-to-one and of the $\sqrt{\cdot}$ to be real. Although the above integrals have closed-form solutions, we omit them due to lack of space. 
By solving the above integrals, we are able to obtain $\bar{A}^{\text{ovlp}}$ and, consequently, $A_{k}^{\text{dom}}$.
}

{
Let $\rho= \frac{L}{A^{\text{tot}}}= \frac{L}{\ell^2}$ denote the density of APs per unit of area.  
Using the results above, the 
{\emph{probability of a nearby AP being exclusively serving the $k$-th UE through its chosen pilot}} is given by:
\begin{align}
    \Psi_k=&\frac{\text{\# APs exclusively serving UE $k$}}{\text{\# APs nearby UE $k$}}\nonumber\\
    =&\dfrac{\rho A^{\text{dom}}_k}{\rho A_k}=1-F_{d_{ki}}(2d^{\lim})\max\left(\overline{\lvert\mathcal{S}_t\rvert}-1,\,0\right)\dfrac{\bar{A}^{\text{ovlp}}}{A_k}.
    \label{eq:Psi}
\end{align}
With $\rho A^{\text{dom}}_k$ (\emph{average \# exclusive pilot-serving APs}) and the above metric, we are able to quantify whether conditions \textit{\textbf{(a)}} and \textit{\textbf{(b)}} that define the concept of spatial separability hold in practice.
}

\subsection{Baseline CF Protocol}\label{sec:BCF}
{From the concept of spatial separability, we propose the BCF RA protocol, consisting of the two steps: \textit{i}) UEs in $\mathcal{K}$ transmit their randomly chosen pilots; \textit{ii}) UEs are only admitted by the network if they are spatially separable according to Definition \ref{def:3}. This baseline scheme ignores the SUCRe resolution methodology introduced by Steps 2 and 3. The best performance of BCF is obtained when $L^{\max}=L$ APs, since this condition increases the probability of a UE being spatially separable. This means that all $L$ APs operate serving all $\tau_p$ pilots for the BCF scheme, that is, the user-centric perspective is not entirely exploited, while consequently increasing the energy consumption of the system (all APs remain \emph{operative}).
}

\section{Estimating UL Signal Power of Colliding UEs}\label{sec:estimators}
In this section, we introduce three different {estimators for} the total UL signal power $\alpha_t$ of colliding UEs {that can be implemented in a distributed manner.} {Distributed here means that each UE can estimate $\alpha_t$ based solely on its own knowledge.} Recall that the estimate of $\alpha_t$ is used in {the SUCRe decision process in \eqref{eq:decision}.} 

\subsection{Estimator 1}
The first method is based on the simplifying assumption that different signal powers $\alpha_{lt}$'s are almost equal for $l\in\mathcal{P}_t$. Therefore, we can assume that $\alpha_{lt}$ in \eqref{eq:approx} is independent of the AP index $l$, yielding the following estimate:
\begin{equation}
    \hat{\alpha}^{\text{est1}}_{t,k} = N \left(\dfrac{\sum_{l\in\mathcal{P}_t}\sqrt{q_lp_k}\tau_p\beta_{kl}}{\Re(z_k)}\right)^2 - \sigma^2.
    \label{eq:estimate1}
\end{equation}
However, this estimator is unfeasible because the $k$-th UE does not know $\mathcal{P}_t$. Since {the $k$-th} UE only knows the {set of nearby} APs, $\mathcal{C}_k$, and it is reasonable to expect that $\mathcal{P}_t\cap\mathcal{C}_k\neq\emptyset$ (Remark \ref{remark:similarity}), a heuristic way to turn the above estimator into a feasible one {consists in adopting}:
\begin{equation}
    \hat{\alpha}^{\text{est1},\text{approx}}_{t,k} = N \left(\dfrac{\sum_{l'\in\mathcal{C}_k}\sqrt{q_{l'}p_k}\tau_p\beta_{kl'}}{\Re(z_k)}\right)^2 - \sigma^2.
    \label{eq:approx-est1}
\end{equation}
Note that the summation is realized now over $\mathcal{C}_k$. Together with this first approximation, we can use the fact that the $k$-th UE approximately knows part of its own contribution to the estimate $\hat{\alpha}^{\text{est1},\text{approx}}_{t,k}$, which is given by $\gamma_k$ defined in \eqref{eq:gamma_k}. Therefore, the first estimator for $\alpha_t$ in \eqref{eq:alpha_t} is
\begin{equation}
    \underline{\hat{\alpha}}^{\text{est1},\text{approx}}_{t,k} = \max\left( N \left(\dfrac{\sum_{l'\in\mathcal{C}_k}\sqrt{q_{l'}p_k}\tau_p\beta_{kl'}}{\Re(z_k)}\right)^2 - \sigma^2, \gamma_k \right).
    \label{eq:estimator1}
\end{equation} 
Using $\gamma_k$ avoids overly underestimation, since the estimate must be at least in the same order of magnitude of $\gamma_k$.

In addition to $\tilde{z}_k$ in \eqref{eq:approx}, Estimator 1 relies on two other approximations. First, due to the simplifying hypothesis of equal $\alpha_{lt}$'s for $l\in\mathcal{P}_t$, the following condition must be satisfied by the APs in $\mathcal{P}_t$:
\begin{equation}
    \underbrace{\sum_{i\in\mathcal{S}_t}\beta_{i1}}_{\text{AP 1}}=\underbrace{\sum_{i\in\mathcal{S}_t}\beta_{i2}}_{\text{AP 2}}=\dots=\underbrace{\sum_{i\in\mathcal{S}_t}\beta_{i|\mathcal{P}_t|}}_{\text{AP } |\mathcal{P}_t|}.
    \label{eq:estimator1-equality-constraint}
\end{equation}%
Geometrically, this means that, if the colliding UEs transmit with the same UL power $p$, the pilot-serving APs must have the same aggregated effective distances $\sum_{i\in\mathcal{S}_t}d_{il}^{-\zeta}$ to the colliding UEs, as seen in \eqref{eq:geometric-interpretation-pt}. Intuitively, note that this condition is easier to be satisfied when $|\mathcal{P}_t|$ is small, {but unlikely to happen in practical scenarios}. Hence, one can expect that this estimator works better for small collision sizes $|\mathcal{S}_t|$, which requires smaller {number of pilot-serving APs}. Second, due to the approximation of $\mathcal{P}_t$ by $\mathcal{C}_k$ on the UE's side, the following sums need to have the same order of magnitude:
$$
\sum_{l\in\mathcal{P}_t}\sqrt{q_lp_k}\tau_p\beta_{kl}\approx\sum_{l'\in\mathcal{C}_k}\sqrt{q_{l'}p_k}\tau_p\beta_{kl'}\hspace{-2mm} \implies\hspace{-2mm}  \sum_{l\in\mathcal{P}_t}\beta_{kl}\approx\sum_{l'\in\mathcal{C}_k}\beta_{kl'},
$$
{when $p_k$ and $q_l$ are constants.} {Deterioration of estimation performance} occurs when the scaling of the sums are not equal, due to the intrinsic difference between {$\mathcal{P}_t$ and $\mathcal{C}_k$}.
\vspace{-7mm}

\subsection{Estimator 2}
The second method avoids the simplifying assumption made for Estimator 1 of equal $\alpha_{lt}$'s and relies on solving the following optimization problem:
\begin{equation}
    \begin{aligned}
        {\underset{\boldsymbol{\alpha}_t\in\mathbb{R}_{+}^{|\mathcal{P}_t|}}{\mathrm{argmin}}}
        \quad & f({\boldsymbol{\alpha}}_t) = \lVert{\boldsymbol{\alpha}}_t\rVert_1 = \sum_{l\in\mathcal{P}_t}\alpha_{lt},\\
        \textrm{s.t.} \quad & g({\boldsymbol{\alpha}}_t) =  \dfrac{\Re(z_k)}{\sqrt{N}}-\tilde{z}_{k}({\boldsymbol{\alpha}}_t) = 0,
    \end{aligned}
    \label{eq:optz-problem}
\end{equation}
where $\alpha_{lt}$'s for $l\in\mathcal{P}_t$ are organized in vector form as ${\boldsymbol{\alpha}}_t\in\mathbb{R}_{+}^{|\mathcal{P}_t|}=[\alpha_{1t},\alpha_{2t},\dots,\alpha_{|\mathcal{P}_t|t}]^\transp$ and $\tilde{z}_{k}(\boldsymbol{\alpha}_t)$ shows dependence of the approximation with $\boldsymbol{\alpha}_t$. Moreover, $f(\boldsymbol{\alpha}_t):\mathbb{R}_{+}^{|\mathcal{P}_t|}\mapsto\mathbb{R}_{+}$ is a linear objective function and $g(\boldsymbol{\alpha}_t):\mathbb{R}_{+}^{|\mathcal{P}_t|}\mapsto\mathbb{R}_{+}$ is a non-linear equality constraint. The motivation behind formulating the problem above comes from the experimental observation that the approximation $\tilde{z}_k$ often overestimates the true value $\Re(z_k)/\sqrt{N}$. In general, this means that most likely $\tilde{z}_k\geq\Re(z_k)/\sqrt{N}$, implying that $g(\boldsymbol{\alpha}_t)<0$ for finite, small $N$. The problem in \eqref{eq:optz-problem} is thus a way to combat this overestimation of $\tilde{z}_k$ by finding the vector $\boldsymbol{\alpha}_t$ which minimizes the total UL signal power $\alpha_t$ of colliding UEs given that the constraint is supposedly satisfied. The result below gives a closed-form solution for \eqref{eq:optz-problem}.

\begin{theorem}\label{theorem:estimate2}
For $k\in\mathcal{S}_t$, let $\hat{\alpha}^{\text{est2}}_{lt,k}$ denote the estimate of ${\alpha}_{lt}$ made by the $k$-th UE. By using \eqref{eq:approx}, we get that:
\begin{equation}
    \hat{\alpha}^{\text{est2}}_{lt,k}= N \left(\dfrac{\sum_{l'\in\mathcal{P}_t}(\sqrt{q_{l'}p_k}\tau_p\beta_{kl'})^{2/3}}{\Re(z_k)}\right)^2 (\sqrt{q_lp_k}\tau_p\beta_{kl})^{2/3} - \sigma^2,
\end{equation}
for all $l \in \mathcal{P}_t$. Therefore, $\hat{\alpha}^{\text{est2}}_{t,k}=\sum_{l\in\mathcal{P}_t}\hat{\alpha}^{\text{est2}}_{lt,k}$.
\end{theorem}
\begin{proof}
    The proof can be seen in the Appendix.
\end{proof}

However, as in the case of Estimator 1, the estimate obtained in Theorem \ref{theorem:estimate2} is unfeasible due to the fact that $\mathcal{P}_t$ is unknown to the $k$-th UE. As before{, we then have:}
\begin{align}
     \hat{\alpha}^{\text{est2},\text{approx}}_{lt,k}&= N \left(\dfrac{\sum_{l'\in\mathcal{C}_k}\mathrm{cte}_{kl'}}{\Re(z_k)}\right)^2\nonumber \mathrm{cte}_{kl} - \sigma^2,
    \label{eq:approx-est2}
\end{align}
for all $l \in \mathcal{C}_k$ and where $\mathrm{cte}_{kl}=(\sqrt{q_{l}p_k}\tau_p\beta_{kl})^{2/3}$ . Therefore, 
\begin{equation}
    \underline{\hat{\alpha}}^{\text{est2},\text{approx}}_{t,k} = \max\left( \hat{\alpha}^{\text{est2},\text{approx}}_{t,k}, \gamma_k \right),
    \label{eq::estimator2}
\end{equation}
where ${\hat{\alpha}}^{\text{est2},\text{approx}}_{t,k}=\sum_{l'\in\mathcal{C}_k}{\hat{\alpha}}^{\text{est2},\text{approx}}_{l't,k}$. Different from Estimator 1 that relies on two approximations other than $\tilde{z}_k$, Estimator 2 is based only in approximating the sum of the average channel gains over $\mathcal{P}_t$ by $\mathcal{C}_k$. 
\vspace{-4mm}

\subsection{Estimator 3}
The third estimator relies on taking more advantage of the fact that the CPU knows the pilot activity matrix $\tilde{\mathbf{A}}$. From this information and the expression in \eqref{eq:check-pilot-t}, the CPU can obtain an estimate of the total UL signal power of colliding UEs:
\begin{equation}
    \widehat{\alpha_{t}}=\sum_{l\in\mathcal{L}} \max\left(\tilde{a}_{tl} - \sigma^2, 0 \right)=\sum_{l\in\mathcal{L}} \max\left(\dfrac{1}{N}\lVert\mathbf{y}_{lt}\rVert^2_2 - \sigma^2, 0 \right),
    \label{eq:widehat-alpha}
\end{equation}
where $\tilde{a}_{tl}$ is entry $(t,l)$ of $\tilde{\mathbf{A}}$. Then, this estimate can be sent back to {the pilot-serving} APs in $\mathcal{P}_t$. The DL precoded signal in \eqref{eq:DL-precoding} can be re-designed as:
\begin{equation}
   \mathbf{V}_l = \sqrt{q_l}\sum_{t\in\mathcal{T}_l} \dfrac{\mathbf{y}_{lt}}{{\sqrt{N \cdot \widehat{\alpha_{t}}}}} \boldsymbol{\phi}^{\transp}_t. 
   \label{eq:new-DL-precoding}
\end{equation}
Now, the effective DL transmit power $\tilde{q}_{lt}$ per pilot is
\begin{equation}
    \tilde{q}_{lt}=\left(\dfrac{q_l}{N\cdot\widehat{\alpha_{t}}}\right)\lVert\mathbf{y}_{lt}\rVert^{2}_{2}.
    \label{eq:est3-eff-tx-power}
\end{equation}
{Because of the denominator $N\cdot\widehat{\alpha_{t}}$,} it is natural to expect that APs will transmit with less power when adopting the precoding in \eqref{eq:new-DL-precoding} than that in \eqref{eq:DL-precoding}. With the DL precoded signal in \eqref{eq:new-DL-precoding}, the approximation $\tilde{z}_k$ becomes: 
\begin{equation}
    \dfrac{\Re(z_k)}{\sqrt{N}}\approx\tilde{z}_k = \sum_{l\in\mathcal{P}_t}\left(\dfrac{\sqrt{q_lp_k}\tau_p\beta_{kl}}{\sqrt{\widehat{\alpha_{t}}}}\right)= \dfrac{1}{\sqrt{\widehat{\alpha_{t}}}}\sum_{l\in\mathcal{P}_t}\sqrt{q_l p_k}\tau_p\beta_{kl}.
    \label{eq:new-approx}
\end{equation}
{Since the power received from the precoded RA response can be very low due to \eqref{eq:est3-eff-tx-power}, which would make it difficult for the UEs to process the signal, we consider that} the $k$-th UE can combat this low received power by multiplying the received signal $z_k$ by a \emph{compensation factor} $\delta$. This pre-processing operation can be described as:
\begin{equation}
    \underline{\Re(z_k)}=\delta\left(\dfrac{(\Re(z_k)-\sigma)}{\sqrt{N}}\right),
\end{equation}
whose subtraction by $\sigma$ helps in not increasing the magnitude of the noise by a factor of $\delta$.\footnote{Although the compensation factor coherently multiplies a part of the interference and noise components, this method provides gain in terms of energy efficiency as the pilot-serving APs effectively transmit with less power.} {In principle, the compensation factor is re-scaling the received signal appropriately so that the UE can take advantage of the information embedded in $\widehat{\alpha_{t}}$.} One way to obtain such $\delta$ is to characterize the ratio:
\begin{equation}
    \delta=\sqrt{\dfrac{q_l}{{\tilde{q}^{\text{avg}}_{l}}}},
    \label{eq:est3-delta}
\end{equation}
where ${{\tilde{q}^{\text{avg}}_{l}}}$ is the average of $\tilde{q}_{lt}$ defined in \eqref{eq:est3-eff-tx-power} with respect to the number of active pilots, number of operative APs, and channel realizations given a number of inactive users, $|\mathcal{U}|$, and probability of activation, $P_a$. In the remainder, we will characterize how and when the UEs can obtain such information and a suitable value of $\delta$ for a scenario of interest. A trivial estimator of $\alpha_t$ is then:
\begin{equation}
    \underline{\hat{\alpha}}^{\text{est3},\text{approx}}_{t,k} = \max\left(\left(\dfrac{\sum_{l'\in\mathcal{C}_k}\sqrt{q_{l'}p_k}\tau_p\beta_{kl'}}{\underline{\Re(z_k)}}\right)^2, \gamma_k \right),
    \label{eq:estimator3}
\end{equation} 
also adopting the approximation of substituting $\mathcal{P}_t$ by $\mathcal{C}_k$.

The main difference between Estimators 1 and 3 is that now the $\alpha_{lt}$'s are in fact \emph{equal} by construction of the DL precoded signal. This is because the latter uses the common factor $(N\cdot\widehat{\alpha_{t}})^{-1/2}$ in \eqref{eq:new-DL-precoding} to normalize the precoded DL signal of {the pilot-serving APs, $\mathcal{P}_t$. However, this normalization reduces the effective DL transmit power $\tilde{q}_{lt}$, making it necessary for more APs to be designated to serve the $t$-th pilot. For this reason, we introduced} the compensation factor $\delta$ to decrease the number of {pilot-serving} APs, $|\mathcal{P}_t|$.

\begin{remark}\label{remark:estimator-cellular}
    {(Comparison with Ce-SUCRe)} Interestingly, all {three} estimators {above} have a similar expression when considering the cellular case where $L=1$ AP and $N=M$ antennas. In this particular case, the estimators become the same as the first estimation method proposed in \cite{Bjornson2017}: 
    \begin{equation}
        \underline{\hat{\alpha}}^{\text{est},\text{approx}}_{t,k} = \max\left(\dfrac{M{qp_k}\tau^2_p\beta^2_{k}}{(\Re(z_k))^2} - \sigma^2, p_k\tau_p\beta_k \right).
        \label{eq:mmtc:cellular-estimator}
    \end{equation}
    Hence, our CF-SUCRe protocol can be seen as a generalization of the Ce-SUCRe by \cite{Bjornson2017}.
\end{remark}

\begin{remark}\label{remark:dependency-iota-Lmax}
    {(Estimators' Dependency with $\iota$ and $L^{\max}$)} All estimators rely on the approximation of $\mathcal{P}_t$ by $\mathcal{C}_k$. This therefore means that the performance of the estimators also depends on the choice of the parameters $\iota$ and $L^{\max}$, since they control the sizes $|\mathcal{C}_k|$ and $|\mathcal{P}_t|$, respectively.
\end{remark}

\subsection{{Evaluating the Estimators Numerically}}
We are interested in better understanding how the three estimators proposed above operate. We consider a CF-mMIMO network comprised of $L=64$ APs disposed in an 8 $\times$ 8 square grid layout, Fig. \ref{fig:illustration-ccal-pcal}. We fix the number of antennas $N=8$ per AP {because it gives} a reasonable approximation $\tilde{z}_k$ in \eqref{eq:approx}. Table \ref{tab:simulation-parameters} summarizes the simulation parameters. Note that the {APs and} UEs transmit with equal power {$q_l$ and} $p${, respectively}. 

\begin{table}[htp]
    \centering
    \caption{Simulation parameters}
    \label{tab:simulation-parameters}
    \small
    \begin{tabular}{rl}
        \hline
        \bf Parameter & \bf Value \\ \hline
        square length $\ell$ & 400 m\\
        multiplicative power constant $\Omega$ & -30.5 dB\\
        pathloss exponent $\zeta$ & 3.67\\
        noise power $\sigma^2$ & -94 dBm \\
        \# pilots $\tau_p$ & 5 \\
        \# APs $L$ & 64 APs \\
        \# antennas per AP $N$ & 8 \\
        DL transmit power per AP $q_l$ & {$\frac{200}{L}$} mW \\
        UL transmit power $p$ & 100 mW \\
        compensation factor $\delta$ & 8\\ 
        \# BS antennas $M$ & 64 \\
        DL transmit power of BS $q$ & 200 mW \\
        \# setups & 100 \\
        \# channel realizations & 100\\
        \hline
    \end{tabular}
\end{table}

\subsubsection{Evaluating Approximation of $\mathcal{P}_t$ by $\mathcal{C}_k$} 
We are particularly interested in evaluating the approximation $\sum_{l\in\mathcal{P}_t}\beta_{kl}\approx\sum_{l'\in\mathcal{C}_k}\beta_{kl'}$ that makes the proposed estimators viable when considering $\mathcal{C}_k=\check{\mathcal{C}}_k$ (natural-nearby APs). {Hence, we define} the \emph{normalized magnitude difference} (NMD) as:
\begin{equation}
{\mathrm{NMD}_k=\left(\sum_{l\in\mathcal{P}_t}\beta_{kl}-\sum_{l'\in\check{\mathcal{C}}_k}\beta_{kl'}\right)/\sum_{l\in\mathcal{P}_t}\beta_{kl}, \quad \text{for UE } k\in\mathcal{S}_t}.
\end{equation}
This metric quantifies how the approximated value $\sum_{l'\in\check{\mathcal{C}}_k}\beta_{kl'}$ is different from the true value $\sum_{l\in\mathcal{P}_t}\beta_{kl}$ in a relative way. 
It measures the similarity of the sets $\check{\mathcal{C}}_k$ and $\mathcal{P}_t$ (Remark \ref{remark:similarity}) based on the values of their elements; when, $\check{\mathcal{C}}_k=\mathcal{P}_t$, $\mathrm{NMD}_k$ evaluates to zero.

\begin{table*}
    \centering
    \caption{Best parameter pair $(|\mathcal{C}_k|,L^{\max})$ at median for $L=64$ APs and $N=8$ antennas per AP.}
 \label{tab:best-parameters-collision}
 \small
 \resizebox{\textwidth}{!}{%
    \begin{tabular}{c|cccccccccc}
    \begin{tabular}[c]{@{}c@{}}Collision\\ size $|\mathcal{S}_t|$\end{tabular} & 1 & 2 & 3 & 4 & 5 & 6 & 7 & 8 & 9 & 10 \\ \hline
    Est. 1 & (6,6) & (3,3) & (7,4) & (7,6) & (7,7) & (7,8) & (7,8) & (7,10) & (7,10) & (7,12) \\
    Est. 2 & (7,10) & (7,7) & (6,6) & (7,8) & (7,7) & (7,8) & (7,9) & (7,10) & (7,11) & (7,14) \\
    Est. 3$^{\dagger}$ & {(5,5,2.29)} & {(7,4,2.13)} & {(7,6,2.56)} & {(7,6,2.59)} & {(7,7,2.8)} & {(7,7,2.82)} & {(7,9,3.18)} & {(7,9,3.19)} & {(7,10,3.36)} & {(7,12,3.66)}\\
    \hline
    \end{tabular}%
    }
    {\raggedright{$^{\dagger}$Estimator 3 is comprised by a triplet where the last element indicates the approximated value of $\delta$ for the chosen parameters.}\par}
    \vspace{-5mm}
\end{table*}

{Fig.} \ref{fig:sum-over-Pcal-Ccal} shows the average NMD, denoted as $\overline{\mathrm{NMD}}$, which is obtained by averaging out the $\mathrm{NMD}_k$ from colliding UEs in $\mathcal{S}_t$ and several realizations of $\mathcal{S}_t$. From Fig. \ref{fig:avg-mag}, we can see that the difference in the scale of the sums is more substantial for a small number of pilot-serving APs $\lvert\mathcal{P}_t\rvert\leq L^{\max}$. The intuition behind this result is that some of the nearby APs contained in $\check{\mathcal{C}}_k$ will not be selected in the construction of $\mathcal{P}_t$ for small $L^{\max}$. One way to reduce this effect is to vary the size of $\mathcal{C}_k$ through $\iota$ instead of using $\check{\mathcal{C}}_k$. Fig. \ref{fig:avg-prob} reveals three other important insights: {
\textbf{a)} the larger the collision size $|\mathcal{S}_t|$ and the smaller the $L^{\max}$, the more likely nearby APs of colliding UEs will be not selected as pilot-serving APs; \textbf{b)} for a value of $L^{\max}$ large enough, the opposite is true: It becomes more likely that pilot-serving APs will contain non-nearby APs, irrelevant to the UEs. The good thing is that these non-nearby APs have average channel gains with decreasing magnitude with distance regarding a UE, making the approximation still good between the sums;  \textbf{c)} the parameter $L^{\max}$ must be selected carefully such that $\mathbb{P}\{\sum_{l \in {\mathcal{P}}_t} \beta_{kl} > \sum_{l' \in \check{\mathcal{C}}_k} \beta_{kl'}\}$ is close to 50\%, meaning that $\sum_{l\in\mathcal{P}_t}\beta_{kl}\approx\sum_{l'\in\check{\mathcal{C}}_k}\beta_{kl'}$. This ensures a good approximation when replacing $\mathcal{P}_t$ by $\check{\mathcal{C}}_k$ in the estimators. In summary, we must unquestionably expect some bias from the approximation of substituting $\mathcal{P}_t$ and $\mathcal{C}_k$, because it is difficult to match these sets when the collision sizes increases ($\uparrow \mathcal{S}_t$ means $\uparrow L^{\max}$). Hence, this bias increases with the collision size, $|\mathcal{S}_t|$, and can be countered in some way with a right selection of $L^{\max}$. However, the bias is also bounded, since the non-nearby APs contained in $\{\mathcal{P}_t\setminus\mathcal{C}_k\}$ becomes irrelevant to the UE due to increasing distance.
}

\begin{figure}[htp]
    \centering
    \subfloat[\small Avg. NMD using $\mathrm{sign}(\overline{\mathrm{NMD}})\log_{10}(1+\overline{\mathrm{NMD}})$ scale.
    ]{\includegraphics[trim={0mm 5mm 1mm 15mm}, clip, width= .45\linewidth]{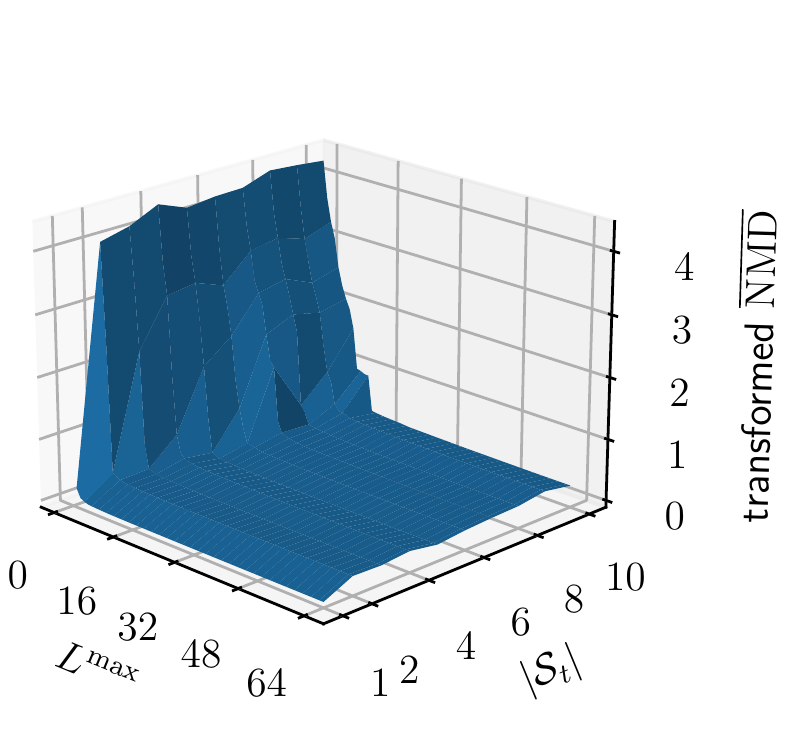}\label{fig:avg-mag}}%
    \qquad
    \subfloat[\small \label{fig:avg-prob} Average evaluation of $\mathbb{P}\{\sum_{l \in {\mathcal{P}}_t} \beta_{kl} > \sum_{l' \in \check{\mathcal{C}}_k} \beta_{kl'}\}$.]{\includegraphics[trim={0mm 2mm 0mm 10mm}, clip, width= .45\linewidth]{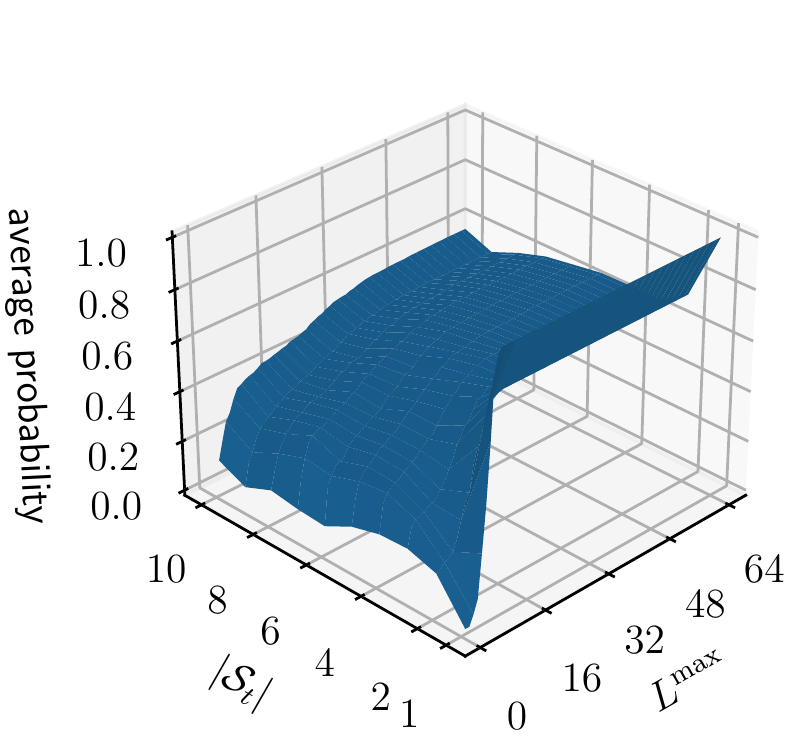}}%
    \caption{\small Evaluating the approximation $\sum_{l \in {\mathcal{P}}_t} \beta_{kl} \approx \sum_{l' \in \check{\mathcal{C}}_k} \beta_{kl'}$ for different collision sizes $|\mathcal{S}_t|$ and choices of $L^{\max}$.}
    \label{fig:sum-over-Pcal-Ccal}
\end{figure}

\subsubsection{Estimator 3 differences}
We now want to understand how the effective DL transmit power $\tilde{q}_{lt}$ in \eqref{eq:est3-eff-tx-power} of the framework used to obtain Estimator 3 changes compared to the more traditional method of transmission defined in \eqref{eq:DL-precoding}. By varying the collision sizes $|\mathcal{S}_t|$ from 1 to 10, UEs' positions, and channel realizations, we found that $\tilde{q}_{lt}$ has an average value of approximately ${{\tilde{q}^{\text{avg}}_{l}}}=0.0489$ mW for $L^{\max}=64$ APs. When compared to the value of $q_l=200/64=3.125$ mW, it is possible to certify that the power of the precoding used in \eqref{eq:new-DL-precoding} is extremely reduced. The compensation factor is approximately $\delta\approx8$ {when $L=64$ APs, $N=8$ antennas per AP, and $L^{\max}=L$. Hence, the value of $\delta$ changes according to the parameter $L^{\max}$.}

\subsubsection{General performance comparison}\label{subsubsec:general-perf-estimators}
Our goal here is to evaluate the estimators under \textbf{i)} different collision sizes $|\mathcal{S}_t|$ and \textbf{ii)} different choices of the parameter pair $(\iota,L^{\max})$ (Remark \ref{remark:dependency-iota-Lmax}). We use the following metrics regarding the $k$-th UE for $k\in\mathcal{S}_t$: \textbf{i)} \emph{normalized estimation bias} (NEB) $b_{t,k}=(\mathbb{E}\{\hat{\alpha}_{t,k}\}-{\alpha}_{t})/{\alpha}_{t}$; and \textbf{ii)} \emph{normalized mean squared error} (NMSE) $\mathrm{NMSE}_{t,k}={\mathbb{E}\{|\hat{\alpha}_{t,k}-\alpha_t|^2\}}/\alpha_t^2$. The expectations are taken with respect to channel realizations. Then, the following simulation routine is adopted: \textbf{1) Setup:} a setup is generated by fixing the number of colliding UEs $|\mathcal{S}_t|$ and by dropping these $|\mathcal{S}_t|$ UEs over the coverage area at random; \textbf{2) Channel Realizations:} several channel realizations are generated for the created setup in 1). For each estimator, we obtain the $\mathrm{NMSE}_{t,k}$ of every $k$-th UE in $\mathcal{S}_t$ and store them; \textbf{3) Statistics:} after the realization of several setups, we evaluate the median value of the stored $\mathrm{NMSE}_{t,k}$'s together with their interquartile range (IQR)\footnote{{The reason why we use median and IQR as metrics is to show that estimators' statistics are not symmetric in general, as they depend on parameters like: $\beta_{kl}$, $|\mathcal{S}_t|$, $L^{\max}$, and $\iota$.}}, which evaluates the variation of the values given the lower and upper quartiles. The above procedure is repeated for all combinations of $(\iota,L^{\max})$. For simplicity, instead of varying $\iota$, we actually parameterize its variation by directly adjusting the number of nearby APs, $|\mathcal{C}_k|$. Therefore, for the considered scenario, suitable ranges are $|\mathcal{C}_k|\in\{1,2,\dots,7\}$ and $L^{\max}\in\mathcal{L}$ (see Fig. \ref{fig:illustration-ccal-pcal}).

First, {by performing exhaustive search,} we seek for the best pair of parameters $(|\mathcal{C}_k|,L^{\max})$ in the sense of obtaining the smallest median ${\mathrm{NMSE}}$ for each one of the evaluated collision sizes $|\mathcal{S}_t|$ and proposed estimators. We report the best parameter pairs $(|\mathcal{C}_k|,L^{\max})$ in Table \ref{tab:best-parameters-collision}. {For Estimator 3, we also reported the best $\delta$ as the last element in a triple.} As expected, the larger the collision size $|\mathcal{S}_t|$, the greater must be the number of pilot-serving APs $L^{\max}$. In addition, the optimal $|\mathcal{C}_k|$ does not vary so much, since $L^{\max}$ was chosen accordingly, making the adjustment of $|\mathcal{C}_k|$ irrelevant in some sense. {However, in practice, it is hard to jointly optimize $|\mathcal{C}_k|$ and $L^{\max}$, since the first is made on the UE side and the second on the CPU side; hence, the choice of $|\mathcal{C}_k|$ is paramount.}

Fig. \ref{fig:collision-size-estimators} shows the best median ${\mathrm{NMSE}}$ obtained for each one of the estimators when adopting the pairs of parameters shown in Table \ref{tab:best-parameters-collision}. As a baseline, we consider a Ce-mMIMO network where a cell-centered BS equipped with $M=64$ antennas serves the colliding UEs using the estimator stated in Remark \ref{remark:estimator-cellular}. Furthermore, in an attempt to be fair in the comparison, the DL transmit power of the BS is $q=L\cdot q_l$, in such a way as to have the same amount of power in the entire coverage area $\sum_{l=1}^{L}q_l=q=200$ mW for both networks. The other parameters of the Ce-mMIMO case follows the ones stated in Table \ref{tab:simulation-parameters}. The results presented in the figure reveal, at median, that among the three CF estimators: \textbf{a)} Estimator 3 is the best in the case of no collision $|\mathcal{S}_t|=1$, \textbf{b)} Estimator 2 works better than the other two on almost all collision sizes. Another interesting fact to note is that CF estimators perform better than the cellular case on almost all evaluated scenarios. The intuition behind this result is the geographic arrangement of APs, which can eventually help to mitigate interference arising from collisions since UEs tend to be far away and, consequently, their nearby APs are distinct. Finally, differences between IQRs are related to the scaling and distribution of the average channel gains $\beta_{kl}$. For the Ce-mMIMO network, as the number of collisions increases, the variability of the estimate is severely affected; since the $\beta_{k}$'s are quite different among edge-UEs (weaker) and center-UEs (stronger), the latter tend to have better estimates than the former. On the other hand, in the CF-mMIMO case, the $\beta_{kl}$ values vary little, as only the $\beta_{kl}$'s in relation to nearby APs have relevance to the estimate and, hence, we have a better uniformity of the estimates. This result reveals one of the main benefits of exploiting user-centric CF networks to solve the RA problem. {However, note that Estimator 3 has the most significant variability among the CF estimators. The explanation is that it becomes challenging to adjust the received signal power through the compensation factor $\delta$ in \eqref{eq:est3-delta} when the collision size increases.}

\begin{figure}[!htbp]
    \vspace{-1mm}
    \centering
    \includegraphics{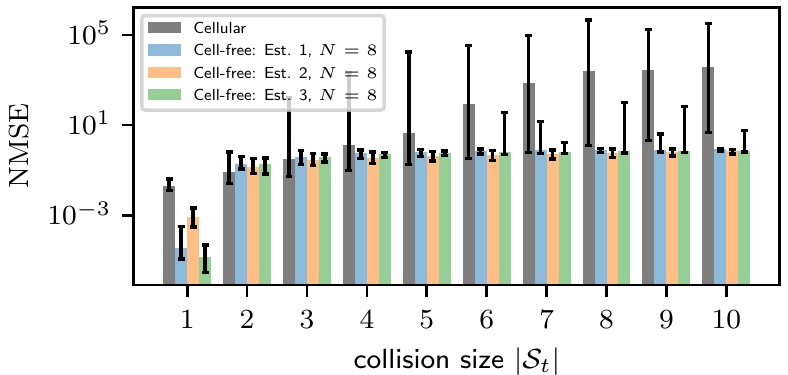}
    \vspace{-1mm}
 \caption{\small Comparison of the performance of the estimators for different collision sizes, $|\mathcal{S}_t|$, in terms of the NMSE metric. The pair of parameters $(|\mathcal{C}_k|,L^{\max})$ is selected according to Table \ref{tab:best-parameters-collision}. The colored bars denote the median values, while the error bars show the IQRs.}
\label{fig:collision-size-estimators}
\end{figure}

In Fig. \ref{fig:antennas-per-ap-estimators}, we evaluate how the performance of the estimators is dependent on the number of antennas $N$ per AP for $|\mathcal{S}_t|=2$ colliding UEs. {For each value of $N$, the best pair of parameters $(|\mathcal{C}_k|,L^{\max})$, including $\delta$ for Estimator 3, was defined by performing an exhaustive search following the same procedure done to generate Table \ref{tab:best-parameters-collision}.} The parameter $N$ is crucial to achieve the channel hardening and favorable propagation effects, which are related to the approximations of: \textbf{i)} $ \frac{1}{N}\lVert\mathbf{y}_{lt}\rVert^2_2$ in \eqref{eq:check-pilot-t} and \eqref{eq:widehat-alpha}; \textbf{ii)} $\tilde{z}_k$ in \eqref{eq:approx} and \eqref{eq:new-approx}. There are three main observations we can learn from the figure: {\textbf{1)} the total number of antennas $L\cdot N$ on the CF-mMIMO network needs to be at least $4\times$ or higher than the number of antennas $M$ at a BS of an "equivalent" Ce-mMIMO network (result in accordance with \cite{Bjornson2017}); \textbf{2)} CF estimators tend to underestimate the true $\alpha_t$ in \eqref{eq:alpha_t} and their estimates are not asymptotically unbiased; a consequence of approximating $\mathcal{P}_t$ by $\mathcal{C}_k$; in general $\sum_{l\in\mathcal{P}_t}\beta_{kl}>\sum_{l'\in\check{\mathcal{C}}_k}\beta_{kl'}$ leading to a negative bias (related to the discussion of Fig. \ref{fig:sum-over-Pcal-Ccal}), and \textbf{3)} CF estimators aim to achieve an NMSE performance floor after a certain increase of $N$; again, this is related to the approximation of $\mathcal{P}_t$ by $\mathcal{C}_k$. Based on these, we introduce the following remark.}

\begin{remark}\label{remark:underestimation}
    {(Underestimation: A consequence of Remark \ref{remark:similarity})} Due to the fact that the UEs tend to underestimate the total UL signal power of colliding UEs $\alpha_t$ in \eqref{eq:alpha_t} due to the approximation of substituting $\mathcal{P}_t$ by $\mathcal{C}_k$, more UEs are expected to declare themselves winners of the contention in \eqref{eq:decision} of Step 3. Fortunately, the concept of spatial separability (Definition \ref{def:3}) can aid the CF network to support the access of more UEs.
\end{remark}

{From Fig. \ref{fig:antennas-per-ap-estimators}, one can observe again that Estimator 2 is the best in terms of overall performance, as expected, based on the fact that it is founded on fewer assumptions than the others.}

\begin{figure}[!htbp]
    \vspace{-1mm}
    \centering
    \subfloat[\label{fig:bias-antennas-per-ap} NEB.]{\includegraphics{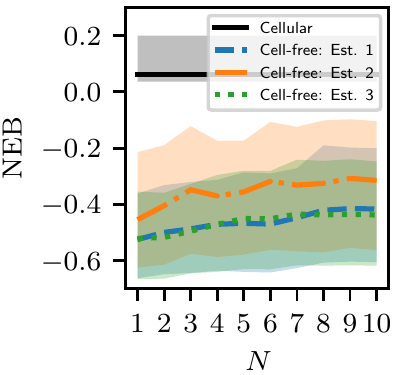}}%
    \quad
    \subfloat[\label{fig:nmse-antennas-per-ap} NMSE.]{\includegraphics{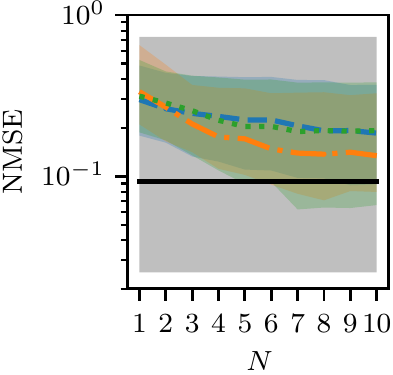}}%
    \vspace{-1mm}
    \caption{\small Evaluation of the performance of the estimators when varying the number of antennas per AP, $N$, for a fixed collision size of $|\mathcal{S}_t|=2$ UEs. The pair of parameters $(|\mathcal{C}_k|,L^{\max})$ is selected according to exhaustive search, together with $\delta$ for Estimator 3. The lines stand for the median values, while the colored regions indicate the IQRs.}
    \label{fig:antennas-per-ap-estimators}
    \vspace{-2mm}
\end{figure}

\subsection{Selecting Parameters in Practice}\label{subsec:selecting-pars}
From the evaluation of the estimators, it was possible to see the importance of selecting the pair of parameters $(|\mathcal{C}_k|,L^{\max})$ for their performance. In this part, we show simple ways to select these parameters in practice.

\subsubsection{Selecting $\iota$ or $|\mathcal{C}_k|$}
For $k\in\mathcal{K}$, the selection of $\iota$ or $|\mathcal{C}_k|$ occurs on the UE's side and it affects the quality of the approximation of substituting $\mathcal{P}_t$ by $\mathcal{C}_k$ made to turn the three estimators feasible in practice. We propose two methods to implement such selection: \textbf{a) Fixed:} {the $k$-th UE} uses the natural set $\check{\mathcal{C}}_k$ to obtain the estimates; \textbf{b) Greedy Flexible:} {the $k$-th UE} obtains $\check{\mathcal{C}}_k$ in Step 0. Then, it evaluates its estimate of $\hat{\alpha}_{t,k}$ by sweeping over all possible sizes of $\mathcal{C}_k$ in a descending order, where $1\leq|\mathcal{C}_k|\leq|\check{\mathcal{C}}_k|$. At each size reduction of $\mathcal{C}_k$, the smallest average channel gain $\beta_{kl'}$ for $l'\in\mathcal{C}_k$ is removed. After sweeping, the $k$-th UE will have a set of estimates $\{\hat{\alpha}^{(|\check{\mathcal{C}}_k|)}_{t,k},\dots,\hat{\alpha}^{(1)}_{t,k}\}$. The $k$-th UE then evaluates the decision in \eqref{eq:decision} for each of the obtained estimates. If one of the decisions indicates that the UE must retransmit, the $k$-th UE decides by $R_k$; otherwise, it chooses $I_k$. The idea is to \emph{greedily} increase the quantity of re-transmissions in order to try to resolve them using the spatial separability principle. For both methods, note that $\sum_{l'\in\mathcal{C}_k}\beta_{kl'}$'s can be computed only once and stored at the $k$-th UE, which is valid in a scenario without much movement. 

\subsubsection{Selecting $L^{\max}$}
The choice of $L^{\max}$ occurs at the CPU and defines the number of pilot-serving APs, $|\mathcal{P}_t|$. For the selection of $L^{\max}$, we propose a \emph{training phase} that assumes that the network has a reliable estimate of the number of inactive UEs, $|\mathcal{U}|$, and the average collision size, $\overline{\rvert\mathcal{S}_t\lvert}$, given a probability of activation, $P_a$. Algorithm \ref{algo:Lmax-selection} describes the procedure to obtain a suitable value for $L^{\max}$, which is independent of the estimator choice. In parallel, the training phase can also be used to obtain the compensation factor $\delta$ by computing the average of {${\tilde{q}_{lt}}$ in \eqref{eq:est3-eff-tx-power}} and broadcasting this value to the UEs. The training phase needs only to be repeated when, for example, the RA performance drops below a selected threshold. Algorithm \ref{algo:Lmax-selection} assumes the use of $\check{\mathcal{C}}_k$.

\begin{algorithm}[htp]
    \centering
    \caption{
    Training phase for selection of $L^{\max}$}
    \label{algo:Lmax-selection}
    \small
    \begin{algorithmic}[1]
        \State \textbf{Input:} Set of inactive UEs: $\mathcal{U}$;\, Probability of activation: $P_a$;\, \# APs: $L$;\, \# antennas per AP: $N$;\, \# RA pilots: $\tau_p$;\, \# random transmission rounds: $R$;\, \# transmission repetitions: $E$ [symbols].
        \State \textbf{Output:} \# pilot-serving APs $L^{\max}$
        \State \textbf{Procedure:}
        \State $T=R\cdot{E}$ \Comment{{\small calculate training duration in symbols}}
        \For{$r\gets{1}$\text{ \textbf{to}} $R$}
            \State generate $\mathcal{K}\subset\mathcal{U}$ given $P_a$
            \State UE $k\in\mathcal{K}$ selects a pilot $t\in\mathcal{T}$ at random
            \For{$e\gets{1}$\text{ \textbf{to}} $E$}
                \State $|\mathcal{K}|$ UEs transmit their pilots as in eq. \eqref{eq:pilot-tx}
                \State CPU calculates and store $\tilde{\mathbf{A}}\in\mathbb{R}_{+}^{\tau_p\times{L}}$ as per \eqref{eq:matrix-A-tilde}
            \EndFor
            \State average out realizations of $\tilde{\mathbf{A}}$ w.r.t. $E$, yielding $\underline{\tilde{\mathbf{A}}}$
            \State $\underline{\tilde{\mathbf{b}}}_t\in\mathbb{R}_{+}^{L}=[\underline{\tilde{\mathbf{A}}}]_{t,:}$ \Comment{{\small $t$-th row of $\underline{\tilde{\mathbf{A}}}$}}
            \For{$t\gets{1}$\text{ \textbf{to} }$\tau_p$} 
                \If{{$t$ is active (was used by a UE in $\mathcal{K}$)}}
                    \State $\epsilon\gets\frac{1}{L}\sum_{l\in\mathcal{L}}\underline{\tilde{{b}}}_{tl}$ \Comment{{\small average threshold}}
                    \State ${L}_t\gets\mathrm{sum}(\underline{\tilde{\mathbf{b}}}_t\geq\epsilon)$ \Comment{{\small element-wise comparison}}
                    \State $\tilde{\tau}_p\gets\tilde{\tau}_p+1$ \Comment{{\small aux. variable initialized as 0}}
                \EndIf
            \EndFor
        \State $L^{\max}_r\gets\frac{1}{\tilde{\tau}_p}\sum_{t=1}^{\tilde{\tau}_p}L_t$
    \EndFor
    \State \Return $L^{\max}\gets\lceil\frac{1}{R}\sum_{r=1}^{R}L^{\max}_r\rceil$
	\end{algorithmic}
\end{algorithm}

\vspace{-4mm}

\section{Numerical Results}\label{sec:num-res}
In this section, we {evaluate the effectiveness of the spatial separability concept. Furthermore, we assess the proposed BCF and CF-SUCRe protocols via access performance and EE evaluations.} The simulation parameters are the same as in Table \ref{tab:simulation-parameters}, unless stated otherwise. In general, we fix the CF design to have $L=64$ APs over the same 8 $\times$ 8 square grid, as in Fig. \ref{fig:illustration-ccal-pcal}. {As a baseline scheme, we naturally consider the Ce-SUCRe by \cite{Bjornson2017} using the estimator from Remark \ref{remark:estimator-cellular}.}

\vspace{2mm}
\noindent {\textbf{Evaluation Metrics.} \emph{Spatial separability:} We evaluate the concept of spatial separability through the two metrics introduced in Subsection \ref{subsec:analysis}: \textbf{i)} the \emph{probability of a nearby AP being exclusively serving the $k$-th UE through its chosen pilot}, $\Psi_k$; \textbf{ii)} the \emph{average number of exclusive-pilot-serving APs}, $\rho A^{\text{dom}}_k$. \emph{Access Performance:} The \emph{average number of access attempts} (ANAA) measures how many accesses on average an inactive UE needs to try such that it can successfully be admitted by the network after the first time the UE starts to be in $\mathcal{K}\subset\mathcal{U}$. We set the maximum NAA a UE can realize to 10 attempts and the probability of reattempt in the next coherence block, if the attempt fails, to 50\%. \emph{Energy Efficiency:} 
The EE on the network's side is evaluated by defining the \emph{total consumed power} (TCP) for the type of network and the RA protocol in question. For the CF-SUCRe, the TCP is
\begin{equation}
    \mathrm{TCP}^{\mathrm{CF}}=\mathrm{ANAA}^{\mathrm{CF}} \cdot \underbrace{(\tau_p+1)}_{\text{\# of DL symbols}}\cdot\underbrace{\vphantom{\tau_{p}}(q_l\cdot\bar{\tau}_{pl})}_{\text{total power per AP}}\cdot\bar{L}, 
\end{equation}
where the superscript "${\mathrm{CF}}$" indicates dependency with the CF-SUCRe protocol, $0\leq\bar{\tau}_{pl}\leq\tau_p$ is the \emph{average number of active pilots per AP}, and $\bar{L}$ is the \emph{average number of operative APs}. The \# DL symbols corresponds to the sum of: \textbf{i)} $\tau_p$ symbols used to respond back to the UEs in Step 2; \textbf{ii)} 1 symbol is used to communicate back with the winning UEs in Step 4. For the other RA protocols, the above metric changes as follows: \textbf{a)} for the CF-SUCRe with Estimator 3, the DL total power of Step 2 depends on $\tilde{q}^{\text{avg}}_{l}$ instead of $q_l$; \textbf{b)} for the Ce-SUCRe, $q_l$ is replaced by $q$, $\bar{\tau}_{pl}$ is independent of $l$, and $\bar{L}=1$; and, \textbf{c)} for the BCF, the \# DL symbols is 1, $\bar{\tau}_{pl}$ is also independent of $l$, and $\bar{L}=L$ APs.
}

{
In Fig. \ref{fig:performance-evaluation}, we evaluate the potential of spatial separability and compare the performance of the RA schemes. Furthermore, we assess the practical methods for selection of parameters $\lvert\mathcal{C}_k\rvert$ and $L^{\max}$ proposed in Subsection \ref{subsec:selecting-pars}. {In general, performance and EE evaluations consider two different bounds for the CF-SUCRe protocol:}
a \textbf{lower bound} (Fig. \ref{fig:performance-lower-bound}) and a \textbf{practical bound} (Fig. \ref{fig:performance-practical}). The {lower bound} considers that the network is able to select the most appropriate $L^{\max}$ based on {exhaustive search and considering that UEs always consider the natural subset of nearby APs, $\check{\mathcal{C}}_k$, to compute estimates, an approach similar to the one performed to obtain Table \ref{tab:best-parameters-collision}}. This establishes the best median performance that one can obtain with the CF-SUCRe for each estimator. While the practical bound uses Algorithm \ref{algo:Lmax-selection} to select $L^{\max}$.
}

\begin{figure*}
    \centering
    \subfloat[][{\footnotesize {{Average} \# exclusive pilot-serving APs, $\rho A^{\text{dom}}_k$ in \eqref{eq:adom}, for $q_{l}=\frac{200}{L}$.}}]{\includegraphics[width=.32\textwidth]{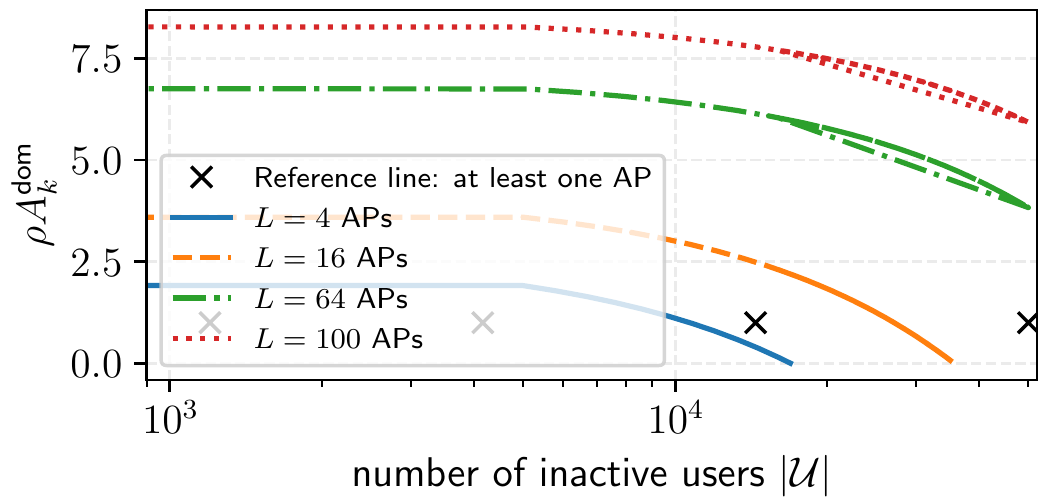}\label{fig:spatial-separability-adom}}
    \hfill
    \subfloat[][{\footnotesize Probability of a nearby AP being exclusively serving the $k$-th UE, $\Psi_k$ in \eqref{eq:Psi}, for $q_{l}=\frac{200}{L}$.}]{\includegraphics[width=.32\textwidth]{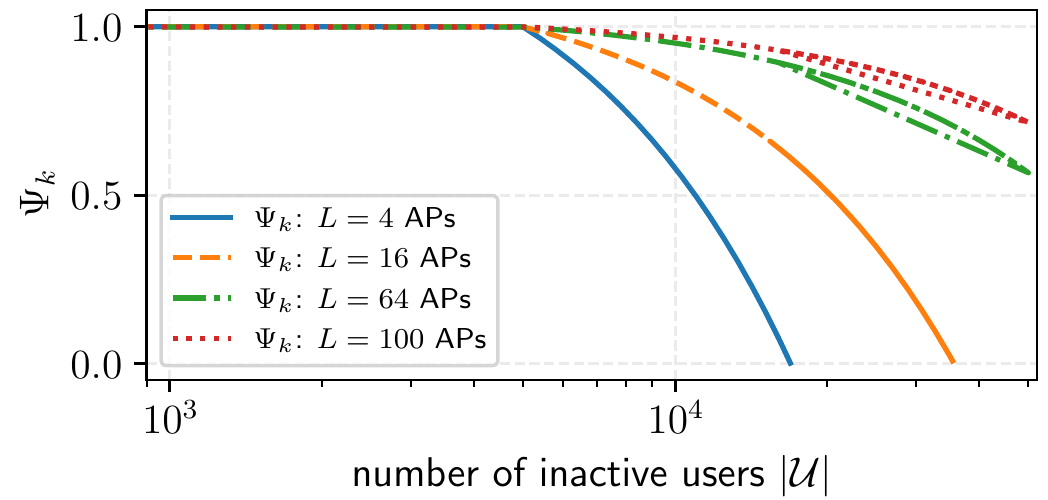}\label{fig:spatial-separability-psi}}
    \hfill
    \subfloat[][{\footnotesize {Lower bound for $L=64$ APs.}}]{\includegraphics[width=.32\textwidth]{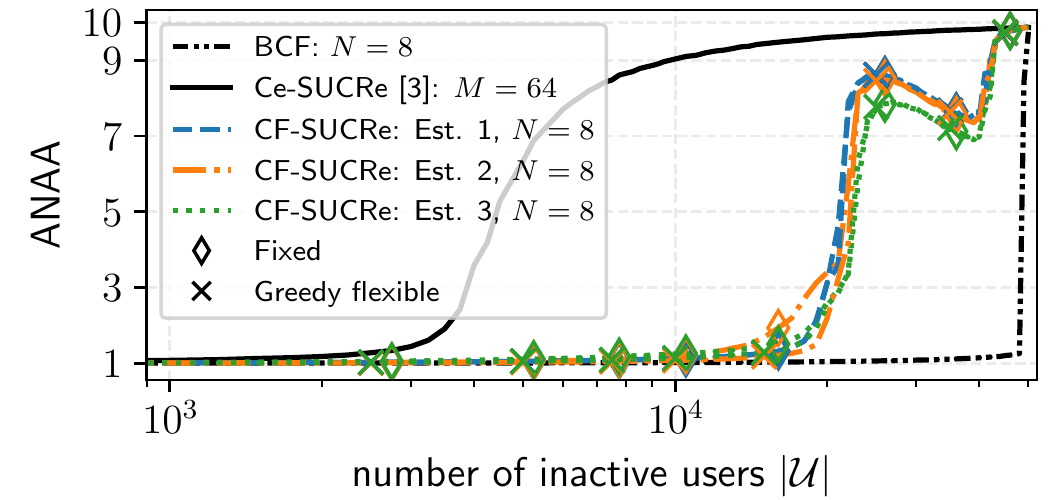}\label{fig:performance-lower-bound}}
    \vskip\baselineskip
    \subfloat[][{\footnotesize {Practical bound (Algorithm \ref{algo:Lmax-selection}) for $L=64$ APs.}}]{\includegraphics[width=.32\textwidth]{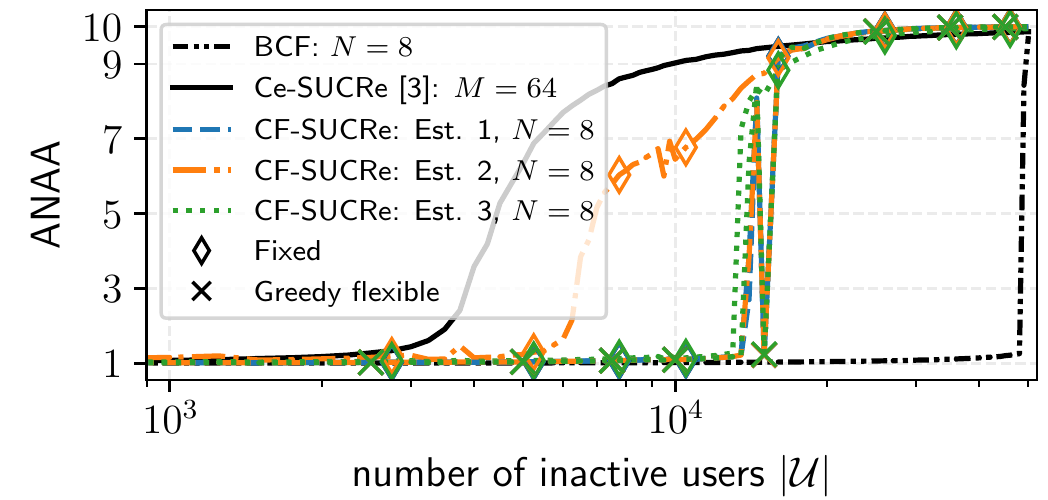}\label{fig:performance-practical}}
    \hfill
    \subfloat[][{\footnotesize {Practical bound (Algorithm \ref{algo:Lmax-selection}) using the greedy flexible method for a fixed $|\mathcal{U}|=10,000$.}}]{\includegraphics[width=0.32\textwidth]{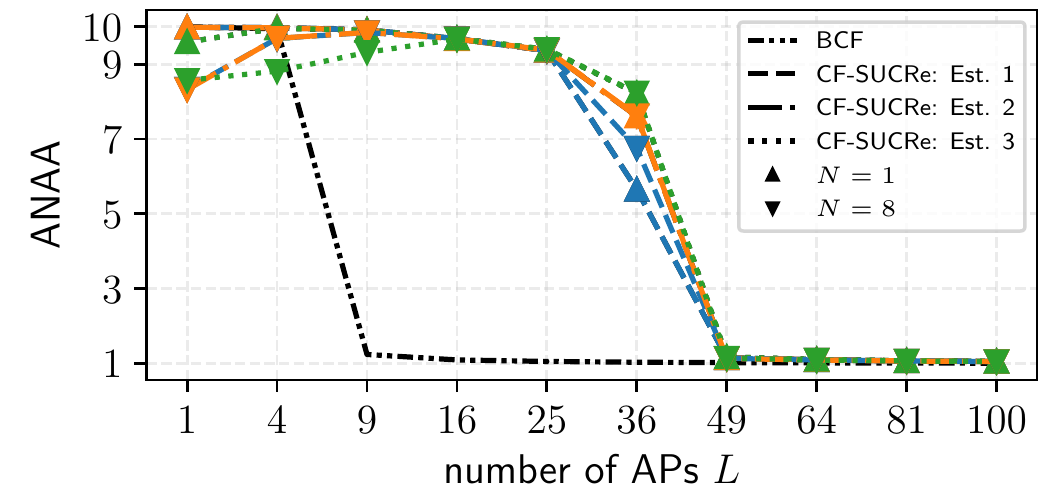} \label{fig:performance-varying}}
    \hfill
    \subfloat[][{\footnotesize {EE evaluation of RA methods. Only the \emph{greedy flexible method} of $|\mathcal{C}_k|$ is considered.}}]{\includegraphics[width=0.32\textwidth]{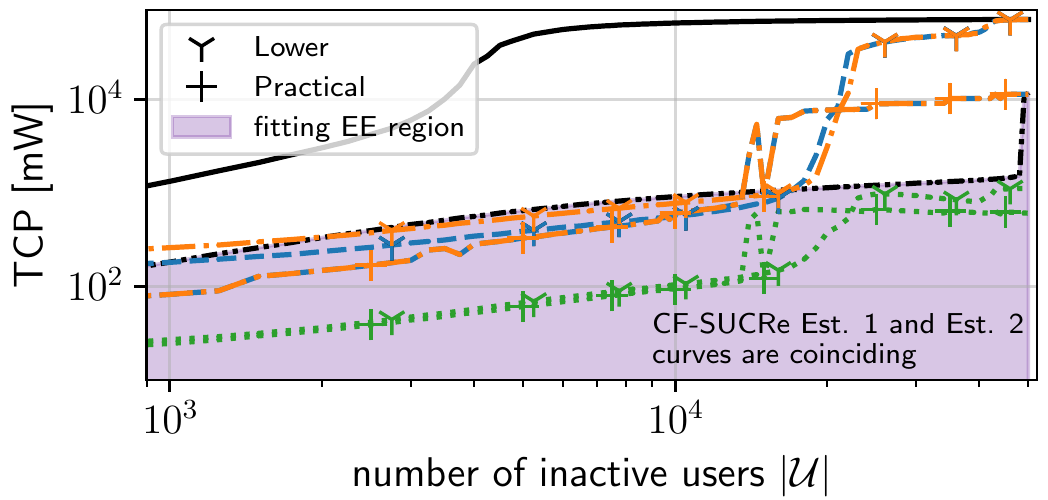} \label{fig:ee-evaluation}}
    
    \caption[]
    {\small {Evaluation of spatial separability, performance{, and EE} of the RA methods: BCF, Ce-SUCRe \cite{Bjornson2017}, and CF-SUCRe. The two selection methods of $|\mathcal{C}_k|$ are assessed for CF-SUCRe: fixed and greedy flexible. In general, the CF-mMIMO network is comprised of $L=64$ APs disposed in a 8 $\times$ 8 square grid layout. Important fixed parameters are: $\tau_p=5$ pilots, $P_a=0.1\%$, maximum \# attempts is 10, probability of reattempts is 50\%. Other parameters are available in Table \ref{tab:simulation-parameters}. {Legends are similar between Figs. (c), (d), and (e)}.}} 
    \label{fig:performance-evaluation}
\end{figure*}

{Figs. \ref{fig:spatial-separability-adom} and \ref{fig:spatial-separability-psi} assess the validity of the conditions \textit{\textbf{(a)}} and \textit{\textbf{(b)}} used to define the concept of spatial separability in Definition \ref{def:3}. The conditions are better attained as $\rho A^{\text{dom}}$ is greater than 1 (at least one exclusive pilot-serving AP) and $\Psi_k$ is closer to 1 {(probability of having a nearby AP exclusively serving the $k$-th UE).} The difference between the curves for different $L$ is due to the change in the DL transmit power $q_l$, which changes the limit distance $d^{\lim}$ in \eqref{eq:limit-distance} {and, consequently, the areas used in Subsection \ref{subsec:analysis}}. The plots reveal that spatial separability improves with $L$ and is achievable, {allowing the network to provide access to multiple UEs simultaneously reusing the same pilot}. {However, by increasing the number of inactive users $|\mathcal{U}|$, the greater the probability of collisions and the size of the collision becomes, eventually reducing the dominant area $A^{\text{dom}}_k$ in \eqref{eq:adom} and, hence, $\Psi_k$ to zero. In particular, for the given set of parameters, the spatial separability starts to drop as from around $|\mathcal{U}|=5,000$ inactive UEs, where the average collision size is $\overline{\lvert{\mathcal{S}}_t\rvert}=({\lvert \mathcal{U}\rvert P_a})/{\tau_p}=1$. Finally, note that the theoretical curves for $L=64$ and $L=100$ APs present oscillations in the range of $|\mathcal{U}|$ greater than $20,000$.} 
}

{
In terms of performance, Figs. \ref{fig:performance-lower-bound} and \ref{fig:performance-practical} show that our CF protocols remarkably outperform the Ce-SUCRe. Based only on the concept of spatial separability, the BCF scheme performs better. The CF-SUCRe performs well and is based on combining two collision resolution strategies: the SUCRe rule in \eqref{eq:decision} and the spatial separability concept in {Definition \ref{def:3}}. Note that the inflection points at which the performance of CF methods abruptly deteriorate are intricately related to when the potential for spatial separability is extinguished in Figs. \ref{fig:spatial-separability-adom} and \ref{fig:spatial-separability-psi}, indicating consistency of the analysis carried out in Subsection \ref{subsec:analysis}. For the case of CF-SUCRe, on one hand, the SUCRe and estimators adversely change the {theoretical} capacity of spatial separability {derived in Subsection \ref{subsec:analysis}}. On the other hand, with the combination of SUCRe and spatial separability, we can improve EE w.r.t. BCF, as shown in the sequel. {Notice that for the CF-SUCRe protocol, the oscillations presented in the ANAA performance around 20,000 inactive UEs can be explained by the same effects observed in Figs. \ref{fig:spatial-separability-adom} and \ref{fig:spatial-separability-psi} together with the use of the SUCRe rule. }
}

{
{Regarding the proposed algorithms, }the "greedy flexible" method for selection of the number of nearby APs, $|\mathcal{C}_k|$, works substantially better than the fixed one in all cases presented. Further, for the CF-SUCRe, the difference between the lower and practical bounds is modest. This implies that, despite being simple, the method for selection of $L^{\max}$ proposed in Algorithm \ref{algo:Lmax-selection} is effective. In general, Est. 2 with greedy flexible achieves the best average performance compared to the cellular case, since it is based on less assumptions. However, it can also be noted that Estimator 3 presents a very good performance in comparison with the other estimators, for both bounds. This is because Estimator 3 outperforms the other two estimators in the case where there is no collision $|\mathcal{S}_t|=1$ and matches well with the performance of the others as the size of the collisions increases, as shown in Fig. \ref{fig:collision-size-estimators}.
}

{
The performance of the BCF and Ce-SUCRe protocols for different number of APs, $L$, and number of antennas per AP, $N$, is evaluated in Fig. \ref{fig:performance-varying} by considering the greedy flexible method (practical bound), and a fixed number of inactive UEs of $|\mathcal{U}|=10,000$. When $L=1$, the system is collapsed to the Ce-SUCRe. For single-antenna APs, notice that ANAA does not always decrease with increasing $L$ for the CF-SUCRe. This is due to the poor performance of the estimators, since the asymptotic approximations ($N\rightarrow\infty$) do not hold. However, for a sufficiently large $N$, ANAA starts to get better and better as we increase $L$. Eventually, CF-SUCRe get the same performance as the BCF
by increasing the number of APs.
}

Fig. \ref{fig:ee-evaluation} evaluates the EE of the RA schemes. Clearly, our proposed CF protocols are more energy efficient than the Ce-SUCRe, mainly because their performance is really superior. But also because not all APs are operative {in the RA phase} for the CF-SUCRe case. This gain becomes clearer when comparing the {CF protocols among themselves}. The BCF protocol uses all $L^{\max}=L=64$ APs to serve the pilots, while the CF-SUCRe does not, which explains why the TCP of CF-SUCRe schemes can be lower than that of the BCF scheme in certain scenarios, even with the BCF having better performance in Fig. \ref{fig:performance-evaluation}. The colored region in Fig. \ref{fig:ee-evaluation} illustrates the EE region of interest in which CF-SUCRe may be a more interesting method than BCF from the energetic point-of-view. Interestingly, Estimator 3 exhibits the best EE gains, due to the reduced effective DL transmit power $\tilde{q}_{lt}$ in \eqref{eq:est3-eff-tx-power}. This fact together with the best performance attained by the Est. 3 in Fig. \ref{fig:performance-lower-bound} can motivate the use of the CF-SUCRe scheme with Est. 3 in a wide range of practical scenarios of interest.

\section{Conclusions}\label{sec:conclusion}
In this work, we considered the extension of the SUCRe method to user-centric CF-mMIMO networks. As we carry out such an extension, we observed that the macro-{diversity} introduced by the user-centric perspective of the CF network can naturally help in the resolution of collisions. With that, we introduced the concept of spatial separability. Then, we proposed two GB RA protocols {for CF systems}: {\bf i)} the BCF that only resolves collisions via spatial separability and {\bf ii)} the CF-SUCRe that combines SUCRe and spatial separability resolutions. For the CF-SUCRe to be implementable, we introduced: \textbf{a)} three estimators to perform SUCRe, \textbf{b)} two methods to select the set of nearby APs (fixed and greedy flexible), and \textbf{c)} one method to select the set of pilot-serving APs (Algorithm \ref{algo:Lmax-selection}). Our numerical results revealed that our CF RA protocols exceedingly outperforms the Ce-SUCRe from \cite{Bjornson2017} under an "equivalent" Ce-mMIMO network. For example, the average EE measured in terms of TPC of CF-SUCRe compared to that of the Ce-SUCRe is on average 340$\times$ lower, reaching up to 800$\times$ in some settings. Moreover, despite the additional overhead, CF-SUCRe performs as well as the BCF for a wide range of inactive UEs {size,} but with an average EE 3$\times$ smaller. {Finally, we evaluate analytically the potential of spatial separability, showing that the technique is achievable in practice. Future research directions can combine the spatial separability principle with other resolution techniques.}

\appendix
We solve \eqref{eq:optz-problem} with the method of Lagrange multipliers. The Lagrange function is:
$
\Gamma(\boldsymbol{\alpha}_t, \mu) = f(\boldsymbol{\alpha}_t) - \mu g(\boldsymbol{\alpha}_t),
$
where $\mu$ is the Lagrange multiplier. We check the constraint feasibility by evaluating: $\nabla_{\boldsymbol{\alpha}_t} f(\boldsymbol{\alpha}_t)=\mu \nabla_{\boldsymbol{\alpha}_t} g(\boldsymbol{\alpha}_t)$. The gradients are:
\begin{equation}
    \nabla_{\boldsymbol{\alpha}_t} f(\boldsymbol{\alpha}_t) \in\mathbb{R}_{+}^{|\mathcal{P}_t|} = \mathbf{1},\nonumber
\end{equation}
\begin{equation}
    [\nabla_{\boldsymbol{\alpha}_t} g(\boldsymbol{\alpha}_t)]_{l} \in\mathbb{R}_{+} = \dfrac{1}{2} \dfrac{\sqrt{q_lp_k}\tau_p\beta_{kl}}{(\alpha_{lt}+\sigma^2)^{3/2}}, \ \forall l \in\mathcal{P}_t.
    \nonumber
\end{equation}
Solving $\nabla_{\boldsymbol{\alpha}_t} f(\boldsymbol{\alpha}_t)=\mu\nabla_{\boldsymbol{\alpha}_t} g(\boldsymbol{\alpha}_t)$ for $\alpha_{lt}$, we get: {$\alpha_{lt} = (\frac{\mu}{2}\sqrt{q_lp_k}\tau_p\beta_{kl})^{2/3} - \sigma^2$.} Substituting this into $g(\boldsymbol{\alpha}_t)$ gives
\begin{equation}
    \left(\dfrac{\mu}{2}\right)^{2/3} = N \left(\dfrac{\sum_{l\in\mathcal{P}_t}(\sqrt{q_lp_k}\tau_p\beta_{kl})^{2/3}}{\Re(z_k)}\right)^2.
    \nonumber
\end{equation}
Plugging the value of $(\mu/2)^{2/3}$ into $\alpha_{lt}$ completes the proof.


\bibliographystyle{IEEEtran}
\bibliography{main.bib}

\end{document}